\documentclass[manuscript,nonacm]{acmart}

\AtBeginDocument{%
  }

\copyrightyear{2026}

\usepackage{graphicx} 
\usepackage{amssymb}
\usepackage{tabularray}
\usepackage{amsmath}
\usepackage{caption}
\usepackage{tikz}
\usepackage{bm}
\usepackage{comment}
\usepackage{hyperref}
\usepackage{multirow}
\usepackage{tabularx}
\usepackage{bbm}
\usepackage[linesnumbered,ruled,vlined]{algorithm2e}
\usepackage{booktabs} 
\usepackage{array}    
\usepackage{geometry} 
\usepackage{pdflscape} 
\usepackage{cleveref}
\usepackage{enumitem}
\usepackage{xcolor}  
\usepackage{colortbl}  
\usepackage{tabularx}
\usetikzlibrary{positioning}
\usetikzlibrary{patterns}

\usepackage{pgfplots}
\usepackage{subcaption}
\usepackage{balance}
\pgfplotsset{compat=1.18}
\usepgfplotslibrary{groupplots}
\pgfplotsset{
    /pgfplots/ybar legend/.style={
        /pgfplots/legend image code/.code={%
            \draw[##1,/tikz/.cd,yshift=-0.25em]
            (0cm,0cm) rectangle (3pt,0.8em);
        },
    },
}

\begin{document}

\title[Realistic Counterfactual Explanations via Denial Constraints]{Realistic Counterfactual Explanations via Denial Constraints}
\newcommand{\todocolor}[1]{\textcolor{cyan}{#1}}
  \newcommand{\todocolorb}[1]{\textcolor{teal}{#1}}
  \newcommand{\todocolorc}[1]{\textcolor{red}{#1}}
  \newcommand{\todocolord}[1]{\textcolor{purple}{#1}}
  \newcommand{\todocolore}[1]{\textcolor{olive}{#1}}  
\newcommand{\daniel}[1]{\todocolor{[[Daniel: #1]]}}
\newcommand{\avia}[1]{\todocolorc{[[Avia: #1]]}}
\newcommand{\ag}[1]{\todocolorb{[[AG: #1]]}}
\newcommand{\nave}[1]{\todocolord{[[Nave: #1]]}}
\newcommand{\sysname}{\texttt{Ours}}
\definecolor{maroon}{cmyk}{0,0.87,0.68,0.32}
\newcommand{\revision}[1]{{\color{blue}#1}}
\newlength{\figheight}
\setlength{\figheight}{4.5cm}

\newcommand{\DiCE}{\textsf{DiCE}}
\newcommand{\PandP}{\textsf{P\&P}}
\newcommand{\PreProc}{\textsf{PreProc}}
\newcommand{\Suspect}{\textsf{Suspect}}
\newcommand{\Vanilla}{\textsf{Vanilla}}
\newcommand{\BestInDataset}{\textsf{Best-in-Dataset}}
\newcommand{\Exhaustive}{\textsf{Exhaustive}}
\newcommand{\PlusDiv}{\textsf{+Div}}
\newcommand{\Ours}{\textsf{Ours}}
\newcommand{\LinOpt}{\textsf{LinOpt}}

\author{Avia Asael}
\affiliation{
  \institution{Tel Aviv University}
  \city{Tel Aviv}
  \country{Israel}
}
\email{aviaasael@mail.tau.ac.il}

\author{Daniel Deutch}
\affiliation{
  \institution{Tel Aviv University}
  \city{Tel Aviv}
  \country{Israel}
}
\email{danielde@post.tau.ac.il}

\author{Nave Frost}
\affiliation{
  \institution{eBay Inc.}
  \city{Netanya}
  \country{Israel}
}
\email{nafrost@ebay.com}

\author{Amir Gilad}
\affiliation{
  \institution{The Hebrew University of Jerusalem}
  \city{Jerusalem}
  \country{Israel}
}
\email{amirg@cs.huji.ac.il}

\begin{abstract}

In the realm of Explainable AI, classification results are often explained via {\em counterfactuals} (CFs for short), which are (ideally small) perturbations to an instance that lead to a change of classification label. Such CFs may serve as explanations for the prediction, pinpointing the features that were important. Existing explainability solutions typically aim at {\em minimizing the distance} of CFs from the original instance so that they are specific to it; and/or {\em maximizing the diversity} of CFs  to cover multiple facets of the reasons underlying the prediction. In this paper, we note that in pursuing these aims, state-of-the-art explainability solutions may (and often do) yield counterfactual explanations that do not correspond to realistic instances. This limits their applicability and usefulness in practice. To remedy this, we    
combine ideas from Explainable AI with ideas from data management. Specifically, we capture realism of CFs via {\em logical constraints} that hold with respect to a dataset of examples (e.g. training set); the class of such constraints that we focus on is that of {\em denial constraints}, extensively studied in the context of relational databases. Algorithmically, we then combine explainable AI solutions to yield CFs, with ideas from data cleaning that we adapt to this unique setting, to transform CFs into realistic ones. Extensive experiments across four datasets validate that our solutions achieve realism with relatively minor compromise in terms of distance and diversity. They further validate that the dedicated optimizations that we have developed to speed up the search for CFs are indeed highly effective.    
\end{abstract}


\maketitle
\newcommand\kddavailabilityurl{https://doi.org/10.5281/zenodo.20271397}
\ifdefempty{\kddavailabilityurl}{}{
\begingroup\small\noindent\raggedright\textbf{Resource Availability:}\\
The source code of this paper has been made publicly available at \url{\kddavailabilityurl}.
\endgroup
}
\section{Introduction}

Counterfactuals (CFs) explain an ML prediction made by a given model for a given instance, through perturbations that yield a change to the assigned label. Intuitively, these perturbations pinpoint the features that were significant for the original prediction. Indeed, there is a large body of research on CFs, both in terms of modeling (which CFs are more useful for explainability) and in terms of algorithms (how to efficiently find highly useful CFs)~\cite{verma2024counterfactual,guidotti2024counterfactual,jiang2024robust}.
  
Consider, for example, the NY Housing dataset~\cite{ny_housing}, a snippet of which appears at the top of Figure~\ref{fig:dataset_dcs}. It includes information on real estate assets, such as their type, number of bedrooms/bathrooms, square footage, sublocality (i.e. neighborhood), etc. A natural classification task is to decide whether an asset's price exceeds $1M$ USD. If the model outputs ``no'' for a particular asset, CFs could identify reasons: perhaps the same asset in a different neighborhood would exceed $1M$ USD? Perhaps a slightly larger square footage would pass the bar? Etc. Among many applications, the availability of such explanations promotes transparency and accountability, and allows identifying bugs and biases in the model.~\cite{wachter2017counterfactual,mothilal2020explaining,verma2024counterfactual}   

To be useful, CFs should be {\em in the proximity} of the original instance: continuing our example, price estimates for very dissimilar assets are uninformative. They should also ideally be {\em diverse}: as exemplified above, there are many possible reasons for the classification, and diverse CFs could cover much of these grounds. Furthermore, they should be {\em realistic}: an explanation stating, for a Manhattan condo, that having $10$ bedrooms would push the price over $1M$ USD is not very useful, if there are no 10-bedroom condos in Manhattan. 

Several concrete notions for capturing proximity (i.e. distance functions) have been proposed, and the problem of finding CFs that optimize them  has been extensively investigated~\cite{wachter2017counterfactual,dandl2020multi,mothilal2020explaining,schleich2021geco}. This in particular includes standard measures such as $L0$ (number of features that have changed) or $L1$ (magnitude of change), where a common variant normalizes per-feature change magnitudes before aggregation. Likewise, the problem of finding diverse CFs has been extensively studied~\cite{mothilal2020explaining,dandl2020multi}. 

By contrast, capturing realism is much more subtle, and here there is a wide range of approaches ranging from restricting CFs to be members of an explicitly given dataset (e.g. the training set) to learning a complex model of what is realistic (see Section \ref{sec:related} for further discussion). Recent work \cite{ben2024cafa}, in a different context of adversarial Machine Learning, has proposed to capture realism via {\em database constraints}, and specifically {\em Denial Constraints} (DCs)~\cite{ChomickiM05}. Originally introduced to capture data integrity, these logical constraints were shown in \cite{ben2024cafa} to strike a good balance between allowing for inference and correctly capturing realism. To illustrate, Table \ref{tab:example_constraints} shows example DCs for 3 standard ML benchmarks. The first two are {\em unary} DCs: they capture constraints, in the form of negated conjunctions, over individual tuples/instances, e.g. that no individual younger than 17 holds a Doctorate, or that Manhattan apartments do not have 10 or more bedrooms. The third is an example of a {\em binary} DC (universally quantified over pairs of tuples, rather than a single tuple), reflecting in this case that tax child exemption laws are consistent for individuals in every state. To capture realism of individual instances (in our case, of CFs), we can think of such binary DCs as a compact representation of many unary ones: given an instance, it will be realistic only if it matches the child exemption status of individuals with the same characteristics in the same state. Using DCs as our basis for capturing realism also has the advantage that their inference/mining from datasets has been extensively studied, and out-of-the-box solutions such as FastADC \cite{xiao2022fast} may be employed.

Our work focuses on \emph{tabular} data, the primary application domain for CF explanations~\cite{verma2024counterfactual,guidotti2024counterfactual} and the setting of all major CF methods. A separate line of work studies CFs for images~\cite{goyal2019counterfactual,jeanneret2022diffusion}, where realism is captured via generative or natural-image priors; our constraint-based notion is specific to relational data, and extending it to unstructured domains is left to future work.

In this paper, we study for the first time the problem of optimizing distance and diversity of CFs {\em while constraining them to adhere to DCs}. This combination has not been studied before (see Section \ref{sec:related} for related work overview); in particular, state-of-the-art solutions for diverse CFs such as DiCE \cite{mothilal2020explaining} are agnostic to constraints and indeed often yield CFs that violate them (as we demonstrate in Section~\ref{sec:experiments}). 

Algorithmically, finding realistic, diverse, and close CFs requires addressing several layers of complexity. First, even without realism constraints, generating optimal CFs is intractable unless $P=NP$ for common model families such as Neural Networks and Random Forests \cite{deutch2019constraints}. Second, we show that projecting a tuple onto the space of realistic tuples (namely those consistent with a given set of DCs) is also intractable unless $P=NP$, even for unary DCs. Third, realism may, and often does, conflict with both the flip-of-label and with the diversity aims.

Our algorithmic solution is then a perturb-and-project framework that repeatedly iterates between {\em perturbation}, namely constraints agnostic CF generation (we incorporate DiCE~\cite{mothilal2020explaining}, a state-of-the-art method for diverse CFs) and {\em projection} onto the space of DC-consistent tuples, which is essentially an optimization-under-constraints problem. A difficulty is that a straightforward use of optimizers such as SMT solvers~\cite{de2008z3} for the latter problem is prohibitively inefficient. This is due to two unique features in our setting: (1) we need to invoke the solver many times in a single execution of the CF-search solution. This is both because we look for multiple diverse CFs and because of the above mentioned conflict between realism and flip-of-label, requiring multiple iterations; (2) Binary DCs must be checked against each relevant tuple in the database, resulting in a large number of assertions for the solver to evaluate.     
We address this through several dedicated optimizations that we have developed: {\em pre-processing and caching} of solver instances; a dedicated constraint-pruning technique, inspired by works in data cleaning~\cite{chu2013holistic,rekatsinas2017holoclean}, that identifies the relevant subset of database tuples ({\em suspect set}); and the incorporation of explicit {\em diversity-based constraints} to avoid redundant projections. The framework is \emph{model-agnostic}: the projection step, our core contribution, treats the classifier as a black-box labeling oracle.

We have conducted an extensive experimental study over four standard benchmarks (Adult, NY Housing, Tax, and Census). Our results show that 55.9--100\% of the CFs produced by DiCE violate Denial Constraints, while our solutions achieve zero violations, namely all CFs that we generate are realistic. This realism comes at only a modest cost: proximity and diversity degrade by under 10\% and ${\sim}3\%$ respectively on most datasets, compared to DiCE CFs. In terms of execution time, the constrained problem is naturally hard, and while a direct SMT-solver formulation is the natural exact approach, it does not scale with over 18{,}000 seconds for a single CF projection on Census. Our optimizations reduce that by up to $63\times$ while preserving the same guarantees.

\begin{table}[t]
    \centering
    \caption{Example Denial Constraints}
    \label{tab:example_constraints}
    \resizebox{\columnwidth}{!}{%
    \begin{tabular}{|l|p{8cm}|}
      \hline
      \textbf{Dataset} & \textbf{Example Denial Constraint} \\ 
      \hline
      Adult & $\neg\{t_0.\text{age} < 17 \land t_0.\text{education} = \text{``Doctorate''}\}$ \\
      \hline
      NY & $\neg\{t_0.\text{sublocality} = \text{``Manhattan''} \land t_0.\text{beds} \geq 10\}$ \\
      \hline
      Tax & $\neg\{t_0.\text{State} = t_1.\text{State} \land t_0.\text{HasChild} = t_1.\text{HasChild} \land t_0.\text{ChildExemp} \neq t_1.\text{ChildExemp}\}$ \\
      \hline
    \end{tabular}%
    }
\end{table}

\paragraph*{Summary of Main Contributions} Our main contributions are:
\begin{enumerate}
\item We formalize the problem of generating diverse counterfactual explanations that adhere to Denial Constraints, combining for the first time the goals of proximity, diversity, and DC-based realism (Section~\ref{sec:prelim}).
\item We develop a perturb-and-project framework with multiple solver optimizations, including constraint caching, suspect-set filtering, and diversity-aware projection (Sections~\ref{sec:probdef}--\ref{sec:project}).
\item We show that the projection subproblem---finding the nearest DC-consistent tuple---is NP-hard (Section~\ref{sec:project}).
\item We conduct an extensive experimental evaluation on four benchmarks, demonstrating that our approach eliminates constraint violations while maintaining proximity and diversity comparable to unconstrained baselines (Section~\ref{sec:experiments}).
\end{enumerate}
\section{Model}
\label{sec:prelim}
We next review basic notions of databases, constraints, and counterfactuals from previous work, leading to our problem statement.
\begin{figure}[t]
\centering
\begin{minipage}{\columnwidth}
\centering
\small
\begin{tabular}{clcccl}
\toprule
ID & Type & Beds & Bath & Sqft & Subloc. \\
\midrule
$d_1$ & Condo & 2 & 2 & 1400 & Manhattan \\
$d_2$ & Condo & 3 & 2 & 704 & Brooklyn \\
$d_3$ & Condo & 2 & 4 & 1568 & Staten\_Island \\
$d_4$ & House & 5 & 6 & 4357 & NY \\
\bottomrule
\end{tabular}
\end{minipage}

\vspace{1em}

\begin{minipage}{\columnwidth}
\small
\textbf{Denial Constraints:}\\[0.5em]
\textbf{DC1}: $\forall t_1, t_2.~ \neg(t_1.\text{Type} = t_2.\text{Type} \wedge t_1.\text{Beds} > t_2.\text{Beds} \wedge t_1.\text{Bath} > t_2.\text{Bath} \wedge t_1.\text{Sqft} < t_2.\text{Sqft})$ \\[0.3em]
\textbf{DC2}: $\forall t.~ \neg(t.\text{Subloc.} = \text{`Manhattan'} \wedge t.\text{Beds} > 4)$ \\[0.3em]
\textbf{DC3}: $\forall t.~ \neg(t.\text{Subloc.} = \text{`Manhattan'} \wedge t.\text{Bath} > 4)$
\end{minipage}
\caption{Sample tuples and denial constraints from the NY Housing Dataset~\cite{ny_housing}. Locality attribute omitted for brevity.}
\label{fig:dataset_dcs}
\end{figure}

\subsection{Databases and Denial Constraints}

A relational schema \cite{abiteboul1995foundations} consists of a {\em relation name}  
$\mathcal{R}$ along with a set of {\em attributes} $\{A_1, \ldots, A_n\}$, where $A_j$ is associated with a domain $dom(A_j)$. A {\em tuple} $t$ maps each attribute $A_j$ to a value $t.A_j \in dom(A_j)$. A schema in general may include multiple relations, but we focus here on the single relation case. 

{\em Integrity Constraints} \cite{abiteboul1995foundations} are often imposed to 
capture data soundness. We focus here on a particularly expressive formalism for describing such constraints, namely Denial Constraints (DCs)~\cite{ChomickiM05,rekatsinas2017holoclean}. DCs are universally quantified first order logic over tuples over a particular relation, reflecting constraints over tuples that may (co-)occur. Specifically, a {\em binary DC} over $sch(\mathcal{R})$ is a first-order formula: 
$\sigma = \forall t_1, t_2 \in \mathcal{R}~ \neg(P_1 \wedge \ldots \wedge P_m)$, 
where each $P_i$ is a predicate of the form $t_1.A \circ t_2.A$ or $t_1.A \circ a$ such that (1) $A \in sch(\mathcal{R})$, and (2) $a \in dom(A)$, and (3) $\circ \in \{=, \neq, >, <, \geq, \leq\}$. Similarly, an {\em unary DC} has the form $\sigma = \forall t \in \mathcal{R}~\neg(P_1 \wedge \ldots \wedge P_m)$ where each $P_i$ has the form $t.A \circ a$ with $a$ and $\circ$ as before. For a given binary DC $dc=\forall t_1, t_2 \in \mathcal{R}~ \neg(P_1 \wedge \ldots \wedge P_m)$ and a pair of tuples $t,t'$ we use $dc[t,t']$ to denote the instantiation of $\neg(P_1 \wedge \ldots \wedge P_m)$ with values from $t$ and $t'$, and we then say $t,t' \models dc$ if $dc[t,t']$ evaluates to $true$ and $t,t' \not\models dc$ otherwise. We similarly define $t \models dc$ and $t \not\models dc$ for unary DCs. We overload notation for sets of DCs (all must hold) and for relations (every tuple/pair must satisfy the DCs).  

\begin{example}
\label{ex:dataset_dcs}
The NY housing dataset~\cite{ny_housing} includes information on properties in New York and is commonly used as a benchmark in ML research, where the prediction goal is to decide whether a property is valued at over $\$1M$. Figure~\ref{fig:dataset_dcs} shows the dataset schema along with $4$ sample tuples. It also shows three example DCs with respect to it: DC1 is a binary DC intuitively stating that for properties of the same type, having more bedrooms \emph{and} more bathrooms implies at least as much square footage. DC2 and DC3 are unary DCs reflecting Manhattan's physical space constraints, namely properties located there do not have more than 4 bedrooms or bathrooms. Indeed, every pair of tuples shown in Figure~\ref{fig:dataset_dcs} satisfies \text{DC1}, and every individual tuple satisfies \text{DC2} and \text{DC3}. We will later (Example~\ref{ex:cf_intro}) show examples of cases where DCs are not satisfied.
\end{example}

\subsection{Counterfactual Explanations}
\begin{table}[t]
\caption{Counterfactuals for tuple $t$. Changed values shown with arrows. $p_{1-3}$ are generated by DiCE perturbation; Ours$_{1-3}$ are our realistic results after projection.}
\label{tab:comparison}
\small
\centering
\begin{tabular}{lcccccc}
\toprule
ID & Type & Beds & Bath & Sqft &  Subloc. & Violation \\
\midrule
$t$ & Condo & 1 & 1 & 679 & Manhattan & --- \\
\midrule
Ours$_1$ & Condo & 1 & $\rightarrow$3 & $\rightarrow$1568 & Manhattan & --- \\
Ours$_2$ & Condo & $\rightarrow$4 & 1 & $\rightarrow$2365 & Manhattan & --- \\
Ours$_3$ & Condo & $\rightarrow$4 & $\rightarrow$2 & $\rightarrow$3075 & Manhattan & --- \\
\midrule
\rowcolor{maroon!20} $p_1$ & Condo & 1 & 1 & $\rightarrow$1750 & Manhattan & DC1 \\
$p_2$ & Condo & $\rightarrow$4 & 1 & $\rightarrow$2365 & Manhattan & --- \\
\rowcolor{maroon!20} $p_3$ & Condo & $\rightarrow$6 & $\rightarrow$3 & $\rightarrow$4103 & Manhattan & DC1, DC2 \\
\bottomrule
\end{tabular}
\end{table}

Given a classifier $M$ and a tuple $t$ with label $M(t)$, a counterfactual (CF) is a tuple $cf$ such that $M(cf) \neq M(t)$~\cite{wachter2017counterfactual}. Intuitively, a CF identifies (preferably small) changes that flip the model's prediction, revealing influential features. Beyond flipping classification, CFs should satisfy proximity (minimal distance), diversity (varied explanations), and realism~\cite{mothilal2020explaining,dandl2020multi}. We next overview specific measures/notions to capture each desideratum.

\paragraph{Proximity (Distance Functions)} To capture proximity, one needs to define a distance function over tuples. Following common practice~\cite{wachter2017counterfactual,mothilal2020explaining,dandl2020multi}, we separately quantify distance  for categorical and numeric attributes and then discuss their combination. For categorical attributes, the common practice is to count the number of changes, namely to use $L_0$: 
\begin{equation}
\label{eq:cat_dist}
dist_{cat}(t,t') = \sum_{A \in attr_{cat}} \mathbbm{1}(t.A \neq t'.A)
\end{equation}
For numeric attributes, by contrast, the magnitude of change should be accounted for. A common way to normalize is based on the Median Absolute Deviation (MAD) ~\cite{mothilal2020explaining, wachter2017counterfactual}, which resembles standard deviation but is considered more robust ~\cite{leys2013detecting}:
\begin{equation}
\small
MAD(A,\mathcal{R}) = \text{median}(\{|t.A - \text{median}(\{t.A : t \in \mathcal{R}\})| \mid t \in \mathcal{R}\})
\label{eq:mad}
\end{equation}

The numeric distance is then defined as:
\begin{equation}
\label{eq:dist_num}
dist_{num}(t,t') = \sum_{A \in attr_{num}}\frac{|t.A - t'.A|}{MAD(A,\mathcal{R})}
\end{equation}

\begin{example}
\label{ex:dist}
Consider tuple $t$ from Table~\ref{tab:comparison} and tuple $d_4$ from Figure~\ref{fig:dataset_dcs}. First, $dist_{cat}(t, d_4) = 2$ since both Type (Condo vs.\ House) and Sublocality (Manhattan vs.\ NY) differ. To compute numeric distance, we first compute the MAD values (over the dataset in Figure~\ref{fig:dataset_dcs}) to be used to normalize. We have $\text{MAD(Beds)} = 1.0$, $\text{MAD(Bath)} = 1.0$, and $\text{MAD(Sqft)} = 608.5$; the latter indicates that square feet varies a lot.  We then have:
\[ dist_{num}(t, d_4) = \tfrac{|1-5|}{1.0} + \tfrac{|1-6|}{1.0} + \tfrac{|679-4357|}{608.5} = 4 + 5 + 6.04 = 15.04 \]
\end{example}

The numeric and categorical distance may then be aggregated in different ways to form an overall distance function over tuples. A common choice is to simply use their (possibly weighted) sum as done e.g. in the implementation of \cite{mothilal2020explaining}:
\begin{equation}
\label{eq:dist_agg}
dist_{agg}(t,t') = dist_{cat}(t,t') + dist_{num}(t,t')
\end{equation}
We will also need to use a distance between a set of tuples and a given tuple, and for that we define $dist_{agg}(T,t')=\frac{\sum_{t\in T} dist_{agg}(t,t')}{|T|}$, overloading notation.

Our framework is robust to the choice of distance function and we discuss alternatives in Appendix~\ref{app:distance}.

\paragraph{Diversity} We aim to find multiple CFs that are diverse, thereby covering different reasons for the classification~\cite{wachter2017counterfactual}. A common diversity measure is that of  \emph{determinantal point process (DPP)} diversity~\cite{mothilal2020explaining}. Let $C = \{t'_1, \ldots, t'_k\}$ be a set of CFs, and let $K$ be a matrix with entries $K_{i,j} = \frac{1}{1+dist_{agg}(t'_i,t'_j)}$. We define: 
\begin{equation}
div(C) = det(K)
\label{eq:dpp}    
\end{equation}

Intuitively, when CFs are very similar to each other, namely $dist_{agg}(t'_i,t'_j)$ values are typically small for $i \neq j$, the off-diagonal entries $K_{i,j}$ are close to 1 and the matrix has low determinant. Conversely, when they are very diverse (namely large distances), off-diagonal entries approach 0, and $K$ approaches the identity matrix with $det(K) \rightarrow 1$ \cite{kulesza2012determinantal}.
\begin{example}
\label{ex:diversity}
Consider two sets of tuples corresponding to the NY-housing schema: 
\begin{align*}
C_1 &= \{(\text{Condo}, 2, 2, 1300, \text{Queens}), (\text{Condo}, 3, 2, 1200, \text{Queens})\}\\
C_2 &= \{(\text{Condo}, 1, 1, 679, \text{Manhattan}), (\text{House}, 3, 2, 1824, \text{Brooklyn})\}
\end{align*}
In $C_1$, tuples have the same values in the type, bath, and sublocality attributes, differing only slightly in the beds and sqft attributes. Conversely, in $C_2$, tuples differ across all attributes. Indeed, $div(C_1) = 0.786 < div(C_2) = 0.983$.
\end{example}

\paragraph{Scoring function} Since distance and diversity are often negatively correlated, a common practice~\cite{mothilal2020explaining,dandl2020multi} is to balance them via a scoring function. Let $w_1, w_2 \geq 0$ with $w_1 + w_2 = 1$ be weight parameters. Given a tuple $t$ and a set $\mathcal{C}$ of CFs for it (w.r.t. a given model) we define:
\begin{equation}
score_t(\mathcal{C}) = w_1 \cdot div(\mathcal{C}) - w_2 \cdot dist_{agg}(\mathcal{C}, t)
\label{eq:score}
\end{equation}

\paragraph{Constraints} We impose two kinds of hard constraints on CFs. The first is that the CFs do not modify particular (pre-defined) attributes called {\em immutable attributes}~\cite{verma2024counterfactual,ustun2019actionable}, thereby focusing the search on desirable attribute modifications. The second type of constraints captures realism: the resulting CF instance should be consistent with a given database and a given set of DCs, namely that its inclusion in the database would not violate any DCs.

\begin{example}
\label{ex:cf_intro}
Consider a classifier $M$ predicting whether a property exceeds \$1M, and the tuple $t$ from Table \ref{tab:comparison}. Assume that $M(t)=0$, namely $M$ assigns label $0$ to the corresponding instance. Further assume that $p_1,p_2,p_3$ in Table \ref{tab:comparison} capture 3 CFs w.r.t. $M$ and $t$, namely $M(p_i)=1 \neq M(t)$. In terms of constraints, we impose the following: (1) sublocality is immutable, reflecting that we wish to examine the effect of other attributes while keeping the neighborhood intact; (2) the resulting CFs should be consistent with DC1 and DC2 with respect to the dataset in Figure \ref{fig:dataset_dcs}. The CF $p_1$, for instance, fails to adhere to these constraints: it is inconsistent with  both $d_1$ and $d_2$ with respect to $DC1$; intuitively, there is a mismatch between its number of rooms and square footage information. $p_3$ is invalid in a similar way based on $DC1$, and another reason for its invalidity is captured by $DC2$: it corresponds to a property with $6$ bedrooms whereas $DC2$ captures a constraint that no property may have more than 4 bedrooms.
\end{example}

\paragraph*{Problem Statement} We next ``put it all together'' to obtain the problem statement to be studied in the sequel. For ease of presentation, we focus on binary classification but the construction is generalizable to the multi-class setting where one defines CFs with respect to a target label. We then define the problem as follows:
\begin{definition}[Diverse Realistic Counterfactuals]\label{def:feasible-cf}
Given a relation $\mathcal{R}$ over a schema $\mathcal{S} = \{A_1, \ldots, A_m\}$, a classifier $M$, a set of DCs $\Sigma$ such that $\mathcal{R} \models \Sigma$, a subset $\mathcal{I} \subseteq \mathcal{S}$ of \emph{immutable attributes}, a natural number $k$, and a tuple $t \in \mathcal{R}$, we aim to find a set $\mathcal{C}$ of size $k$ that is a solution to:
\begin{align*}
\max_{\mathcal{C}} \quad & score_t(\mathcal{C}) \\
\text{s.t.} \quad & \forall cf \in \mathcal{C}:~ cf \in dom(A_1) \times \ldots \times dom(A_m),~ M(cf) \neq M(t) \\
& \forall cf \in \mathcal{C}:~ \mathcal{R} \cup \{cf\} \models \Sigma,~ \forall A \in \mathcal{I}:~ cf.A = t.A
\end{align*}
\end{definition}

Intuitively, we aim to find a set of $k$ CFs for a specific tuple $t$ such that each one is within the domain of the schema attributes, it has identical immutable attributes to $t$, and when added to the database, each CF does not cause DC violations. 
This set of CFs should maximize the score function across all sets that adhere to these criteria. 

\begin{example}
\label{ex:score}
Consider the set $\mathcal{C} = \{\text{Ours}_{1-3}\}$ from Table~\ref{tab:comparison}. Each CF satisfies the constraints: immutable attributes are preserved (Sublocality remains Manhattan), and $\mathcal{R} \cup \{cf\} \models DC_{1-2}$ for each $cf \in \mathcal{C}$. With $dist_{agg}(\mathcal{C}, t) = 5.72$, $div(\mathcal{C}) = 0.88$, and weights $w_1 = \frac{2}{3}$, $w_2 = \frac{1}{3}$ (a standard choice following~\cite{mothilal2020explaining}), the overall score is $score_t(\mathcal{C}) = \frac{2}{3} \cdot 0.88 - \frac{1}{3} \cdot 5.72 = -1.32$.
\end{example}
\section{Generating Realistic Counterfactuals}
\label{sec:probdef}


Algorithm~\ref{alg:perturb-project} presents our framework for finding realistic and diverse CFs. We maintain the collection of CFs found so far (Line \ref{line:init}) and a candidate queue (Line \ref{line:list}). We iterate (Line \ref{line:while}) until sufficiently many CFs are found. In each iteration, we pop a candidate (Line \ref{line:popfirst}) and invoke {\tt PERTURB} (Line \ref{line:perturb}) to generate $k$ CFs consistent with immutable attributes, ignoring DCs for now. As such, the resulting CFs may violate DCs.  For that, we invoke a {\tt PROJECT} function (Line \ref{line:project}, details below) whose goal is to find a tuple $p'$ in the proximity of each CF $p$. The resulting $p'$ may or may not be a CF. If it is (Line \ref{line:accept}), we add it to the collection of realistic CFs; and otherwise (Line \ref{line:reject}) we add it to the queue for the next iteration. This process may result in more than $k$ CFs and the final step is to greedily choose $k$ CFs out of them (Line \ref{line:select}, details below). Figure \ref{fig:algo_graph}  includes a schematic illustration of this iterative process.         
We next detail the implementation of {\tt PERTURB}, {\tt PROJECT}, and {\tt CHOOSEK-DIVERSE}. {\tt PERTURB} is fairly standard and   {\tt CHOOSEK-DIVERSE} is fairly simple, and we describe both here. By contrast, {\tt PROJECT} requires new development and is described at a high-level here and in more detail in the next section.


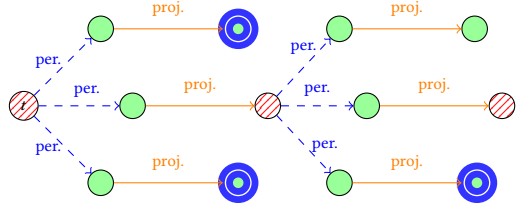
\begin{figure}[t]
\centering
\begin{tikzpicture}[scale=0.80, transform shape, node distance={18mm},
    label0/.style = {draw, circle, minimum size=12pt, font=\footnotesize, fill=red!50, pattern=north east lines, pattern color=red!90},
    label1/.style = {draw, circle, minimum size=12pt, font=\footnotesize, fill=green!40},
    selected/.style = {draw, circle, minimum size=12pt, font=\footnotesize, fill=green!40, draw=blue!80, line width=0.9mm, double, double distance=0.5pt}]
\node[label0] (0) {$t$};
\node[label1] (1) [above right of=0]{};
\node[label1] (2) [right of=0]{};
\node[label1] (3) [below right of=0]{};
\node[selected] (4) [right = of 1]{};
\node[label0] (5) [right = of 2]{};
\node[selected] (6) [right = of 3]{};
\node[label1] (7) [right = 12mm of 4]{};
\node[label1] (8) [right = 12mm of 5]{};
\node[label1] (9) [right = 12mm of 6]{};
\node[label1] (10) [right = of 7]{};
\node[label0] (11) [right = of 8]{};
\node[selected] (12) [right = of 9]{};
\draw[->,blue, dashed] (0) -- node [near start,above=1mm] {\small per.}(1);
\draw[->,blue, dashed] (0) -- node [pos=0.6,above=0.4mm] {\small per.}(2);
\draw[->,blue, dashed] (0) -- node [near start,below=1mm] {\small per.}(3);
\draw[->,orange] (1) -- node [midway,above=1mm] {\small proj.}(4);
\draw[->,orange] (2) -- node [midway,above=1mm] {\small proj.}(5);
\draw[->,orange] (3) -- node [midway,above=1mm] {\small proj.}(6);
\draw[->,blue, dashed] (5) -- node [midway,above=1mm] {\small per.}(7);
\draw[->,blue, dashed] (5) -- node [midway,above=0.5mm] {\small per.}(8);
\draw[->,blue, dashed] (5) -- node [pos = 0.9,above=1.4mm] {\small per.}(9);
\draw[->,orange] (7) -- node [midway,above=1mm] {\small proj.}(10);
\draw[->,orange] (8) -- node [midway,above=1mm] {\small proj.}(11);
\draw[->,orange] (9) -- node [midway,above=1mm] {\small proj.}(12);
\end{tikzpicture}
\caption{Perturb-and-project for CF generation. Hatched nodes: label 0; solid green: label 1; bordered: final realistic CFs. Dashed blue edges: perturbation; solid orange: projection.}
\label{fig:algo_graph}
\end{figure}

\begin{algorithm}[t]
\caption{Perturb-and-Project for Realistic CFs}\label{alg:perturb-project}
\SetAlgoLined
\SetKwInOut{Input}{input}
\SetKwInOut{Output}{output}
\DontPrintSemicolon
\Input{Original tuple $t$, number of CFs $k$, classifier $M$, relation $\mathcal{R}$, set of DCs $\Sigma$,  immutable attributes set $\mathcal{I}$}
\Output{Set $\mathcal{C}$ of $k$ realistic counterfactuals}
$\mathcal{C} \gets \emptyset$\; \label{line:init}
$L \gets [t]$\ \label{line:list} \tcp{FIFO: pop front, append back}
\While{$|\mathcal{C}| < k$ and $L \neq \emptyset$}{ \label{line:while}
    $t' \gets L.\text{popFirst}()$\;\label{line:popfirst}
    $P \gets \textsc{Perturb}(t', M,k,\mathcal{I})$\; \label{line:perturb}
    \For{each $p \in P$}{ \label{line:forloop}
        $p' \gets \textsc{Project}(p, \mathcal{R}, \Sigma, \mathcal{I})$\; \label{line:project}
        \eIf{$M(p') = 1$}{ \label{line:check}
            $\mathcal{C} \gets \mathcal{C} \cup \{p'\}$\; \label{line:accept}
        }{
            $L.\text{append}(p')$\; \label{line:reject}
        }
    }
}
\Return $\textsc{ChooseK-Diverse}(\mathcal{C}, t, k)$\; \label{line:select}
\end{algorithm}

\paragraph{Perturbation Phase (Line~\ref{line:perturb})} The {\tt PERTURB} function is assigned with outputting $k$ diverse CFs in the proximity of a given tuple, which are further consistent with the immutable attributes (it is agnostic to DCs). This sub-problem has been studied in the literature and shown to be NP-hard for models such as Neural Networks and Random Forests (see e.g.~\cite{deutch2019constraints}). On the other hand, effective heuristic solutions have been developed, and a particularly notable such solution is DiCE~\cite{mothilal2020explaining}. DiCE finds CFs for Neural Networks based on backpropagation, jointly optimizing for classification confidence, distance from the input, and diversity. Other solutions are discussed in Section~\ref{sec:related}.  



\begin{example}[Perturbation Phase]
\label{ex:perturb}

We have trained a Neural Network model $M$ over the NY housing dataset and it classifies the tuple $t$ from Table \ref{tab:comparison} to $0$ (predicted price below \$1M). Invoking DiCE over $t$ and asking for $k=3$ CFs, we obtain $p_1$, $p_2$, and $p_3$ from Table~\ref{tab:comparison}. They are CFs for $t$, namely accepted by the model, yet as mentioned above, two of them violate DCs. 
\end{example}

\paragraph{Projection Phase (Line~\ref{line:project})} 
The main novelty of our solution with respect to DiCE lies in the incorporation of the realism requirement, captured through DCs. To this end, the {\tt PROJECT} function is geared towards projecting its input onto the space of realistic tuples, namely, to find a  tuple in the proximity of the input tuple, that further satisfies the constraints. We provide details on our implementation of {\tt PROJECT} in Section \ref{sec:project}, and for now just illustrate its input and output.  
\begin{example}[Projection Phase]
\label{ex:project}
Continuing Example~\ref{ex:perturb} and referring to tuples from Table \ref{tab:comparison}, $p_1$ violates DC1 with respect to $d_1$ and $d_2$ from Figure~\ref{fig:dataset_dcs}, as well as with other tuples in the dataset not shown here for brevity.  Intuitively, for a (1,1,Sqft) condo to satisfy DC1, its Sqft must be at most the smallest Sqft among all condos with beds$>$1 and baths$>$1. This quantity is $704$ and $p_1$ is projected to obtain $(1,1,704)$. This however flips the label back to 0, leading to another iteration of perturb-and-project (see Example \ref{ex:cf_example}). By contrast, $p_2$ and $p_3$ are projected to Ours$_2$ and Ours$_3$ respectively, thereby satisfying the DCs while, in this case, maintaining the ``accept" label. 

\end{example}

\paragraph{Greedy Selection (Line~\ref{line:select})} The iterative process may yield more than $k$ valid candidates, since each DiCE invocation produces $k$ new branches (see Algorithm~\ref{alg:perturb-project} for illustration). The final phase is then to select a subset of size $k$ out of the obtained candidates set. This is done in a greedy fashion: we repeatedly choose the candidate whose inclusion in the set maximizes the overall score (with respect to the score function defined above, aggregating proximity and diversity scores) among all such choices.      
\begin{example}[Complete Pipeline]
\label{ex:cf_example}
Continuing Example \ref{ex:project}, our candidates set $\mathcal{C}=\{Ours_2,Ours_3\}$ and the queue $L=\{(1,1,704)\}$. This single tuple occupying $L$ is popped and fed as input to DiCE, which yields 3 new candidates. Each of them is again projected onto the space of valid tuples, and two of the resulting projections are accepted by the model, namely Ours$_1$=(1,3,1568) and a fourth candidate (4,1,2393). With four realistic CFs, greedy selection operates by proximity and diversity, choosing the best 3 which are $\{Ours_1,Ours_2,Ours_3\}$.
\end{example}

\section{Projecting over Denial Constraints}
\label{sec:project}

The algorithm in Section \ref{sec:probdef} invokes three functions as subroutines, for perturbation, projection, and greedy selection. Implementations for perturbation and greedy selection were described in the previous section, and here we detail the projection algorithm. Specifically, we show that the problem is NP-hard; present a solution based on SMT solvers~\cite{de2008z3}; and provide multiple crucial optimizations.  

\paragraph*{Problem Statement and Intractability} Given a tuple $t$ to project, a set $\Sigma$ of DCs and a relation $\mathcal{R}$, the projection problem is:

\begin{equation}
\begin{aligned}
\min_{t'} \quad & dist_{agg}(t, t') \quad \text{s.t.} \quad \mathcal{R} \cup \{t'\} \models \Sigma, \quad \forall A \in \mathcal{I}: t'.A = t.A
\end{aligned}
\label{eq:projection}
\end{equation}

We show that, unless P=NP, there is no hope for an exact PTIME solution, since even checking for satisfiability of DCs is already intractable. For that, we define {\tt DEC-PROJ} as the problem of deciding, given $\Sigma$ and $\mathcal{R}$ as above, whether there exists a tuple $t'$ such that $\mathcal{R} \cup \{t'\} \models \Sigma$. We show (proof in Appendix~\ref{app:proofs}):
\begin{proposition}
\label{prop:np-hard}
{\tt DEC-PROJ} is NP-hard, even if $\Sigma$ contains only unary DCs.
\end{proposition}

 We consequently focus on solver-based solutions that do not guarantee PTIME convergence but perform well in practice.

\begin{algorithm}[t]
\caption{Vanilla Projection}\label{alg:vanilla-project}
\SetAlgoLined
\SetKwInOut{Input}{input}
\SetKwInOut{Output}{output}
\DontPrintSemicolon
\Input{$t$: tuple to project, $\mathcal{R}$: relation, $\Sigma$: DCs, $\mathcal{I}$: immutable attributes}
\Output{$t'$: projected tuple, or $\bot$ if infeasible}
\tcp{Check for immutable violations}
\For{each $\sigma \in \Sigma$}{ \label{line:v-infeas-start}
    \lIf{$attr(\sigma) \subseteq \mathcal{I}$ \textbf{and} $\exists t_i \in \mathcal{R}: (t, t_i) \not\models \sigma$}{\Return $\bot$} \label{line:v-infeas-end}
}
$opt \gets \text{CreateOptimizer}()$\; \label{line:v-create}
$vars \gets \text{CreateVariables}(sch(\mathcal{R}))$\; \label{line:v-vars}
\tcp{Add DC satisfaction constraints}
\For{each $\sigma \in \Sigma$, each $t_i \in \mathcal{R}$}{ \label{line:v-dc-start}
    $bound \gets \text{InstantiateDC}(\sigma, t_i)$\;
    $opt.\text{add}(\neg \bigwedge_{(A,v,op) \in bound} (vars[A] \circ_{op} v))$\; \label{line:v-dc-end}
}
\tcp{Fix immutable attributes}
\For{each $A \in \mathcal{I}$}{ \label{line:v-immut-start}
    $opt.\text{add}(vars[A] = t.A)$\; \label{line:v-immut-end}
}
\tcp{Set objective and solve}
$opt.\text{minimize}(dist_{agg}(vars, t))$\; \label{line:v-obj}
\lIf{$opt.\text{solve}() = \text{SAT}$}{\Return $opt.\text{getModel}(vars)$} \label{line:v-solve}
\Return $\bot$\;
\end{algorithm}

\paragraph*{``Vanilla'' Solver-Based Solution} Algorithm~\ref{alg:vanilla-project} presents a baseline solution for the projection problem based on an SMT solver. We first check for infeasibility (Lines~\ref{line:v-infeas-start}--\ref{line:v-infeas-end}): if $t$ violates a DC involving only immutable attributes, it cannot be fixed and we return $\bot$ to signal this. Otherwise, the solver is initialized with variables corresponding to the schema (Lines~\ref{line:v-create}--\ref{line:v-vars}). Then, for each DC $\sigma \in \Sigma$ and each tuple $t_i \in \mathcal{R}$, we instantiate $\sigma$ with values from $t_i$ to encode ``forbidden regions'' that the projected tuple must avoid (Lines~\ref{line:v-dc-start}--\ref{line:v-dc-end}). We then create a boolean expression whose form is the negation of the conjunction of these forbidden regions. We further (Lines~\ref{line:v-immut-start}--\ref{line:v-immut-end}) add boolean constraints capturing that immutable attributes must be fixed to their values in $t$, set the objective to distance minimization and run the solver, returning the solution it found or $\bot$ if none was found.

\begin{example}
\label{ex:vanilla}
Consider projecting tuple $p_1$ from Table~\ref{tab:comparison} with immutable attributes $\mathcal{I} = \{\text{Type}, \text{Subloc}\}$. For DC1 and tuple $d_1$ from Figure~\ref{fig:dataset_dcs}, instantiation yields the bound $\{(\text{Type}, \text{Condo}, =), (\text{Beds}, 2, <), (\text{Bath}, 2, <), (\text{Sqft}, 1400, >)\}$. Similar bounds are generated for all $|\mathcal{R}| \cdot |\Sigma|$ tuple-DC pairs, and we conjoin over all of them.
\end{example}

\subsection{Optimizations and Variations}

There are multiple problems slowing down Algorithm~\ref{alg:vanilla-project} and specifically its use in our perturb-and-project framework, as follows:
\begin{itemize}\begin{sloppypar}
\item The framework involves multiple iterations with many {\tt project} invocations, with the same DCs but with different choices of tuple $t$ to project. This means that solver initialization and DCs instantiations are repeated across these invocations.
\end{sloppypar}
\item The number of DC instantiations is $O(|\mathcal{R}| \cdot |\Sigma|)$, one for each tuple in $\mathcal{R}$ and DC in $\Sigma$, which may be very large. 
\item Multiple projections of different tuples may yield the same or similar solutions, which is problematic both in terms of effective search space exploration and in terms of diversity of the final result. 
\end{itemize}
We next present an optimization or variant accounting for each of these points.

\paragraph*{Pre-processing with Optimizer Caching} To account for the repeated initialization step, Algorithm~\ref{alg:pre-processing} constructs and caches optimizer instances keyed by the immutable attribute set $\mathcal{I}$. The key observation is that DC satisfaction constraints (Lines~\ref{line:v-dc-start}--\ref{line:v-dc-end}) depend only on $\mathcal{R}$ and $\Sigma$, not on the specific tuple being projected. Pre-processing thus generates all $O(|\mathcal{R}| \cdot |\Sigma|)$ constraints upfront; once cached, projecting any tuple requires only adding immutable constraints (Lines~\ref{line:v-immut-start}--\ref{line:v-immut-end}) and the distance objective (Line~\ref{line:v-obj}).

\begin{algorithm}[t]
\caption{Pre-processing for Optimizer Construction}\label{alg:pre-processing}
\SetAlgoLined
\SetKwInOut{Input}{input}
\SetKwInOut{Output}{output}
\DontPrintSemicolon
\Input{$\mathcal{R}$: relation, $\Sigma$: DCs, $\mathcal{I}$: immutable attributes, $cache$: global optimizer cache}
\Output{$(opt, vars)$: configured optimizer and variables}
\If{$cache[\mathcal{I}]$ is empty}{
    $opt \gets \text{CreateOptimizer}()$\;
    $vars \gets \text{CreateVariables}(sch(\mathcal{R}))$\;
    \For{each $\sigma \in \Sigma$, each $t_i \in \mathcal{R}$}{
        $bound \gets \text{InstantiateDC}(\sigma, t_i)$\;
        $opt.\text{add}(\neg \bigwedge_{(A,v,op) \in bound} (vars[A] \circ_{op} v))$\;
    }
    $cache[\mathcal{I}] \gets (opt, vars)$\;
}
\Return $cache[\mathcal{I}]$\;
\end{algorithm}


\paragraph*{Suspect Set Filtering} Pre-processing trades upfront cost for universal reuse. But what if we only need CFs for a few tuples, or many tuples share the same immutable attribute values? The next optimization follows the opposite approach: rather than pre-processing all constraints, it builds an optimizer on-demand containing only constraints relevant to the specific tuple's immutable values. This approach is inspired by suspect set techniques used in the different context of constraint-based data cleaning~\cite{chu2013holistic}. 

\begin{definition}
Given tuple $t$, immutable attributes $\mathcal{I}$, relation $\mathcal{R}$, and DCs set $\Sigma$, let $\sigma_{\mathcal{I}}$ denote predicates in $\sigma$ involving only attributes in $\mathcal{I}$. A tuple $t' \in \mathcal{R}$ is a \emph{suspect} w.r.t.\ $t, \mathcal{I}, \sigma$ if $t, t'$ satisfy all predicates in $\sigma_{\mathcal{I}}$. Namely:
$$suspect(t, \mathcal{R}, \sigma, \mathcal{I}) = \{t' \in \mathcal{R}\ |\ \forall_{p \in \sigma_{\mathcal{I}}} (t, t') \models p\}$$
\end{definition}
The key insight is that  if $t$ and $t'$ disagree on an immutable predicate, no modification to mutable attributes can cause a violation and we thus need no constraints corresponding to $t'$. To implement this optimization, we modify Lines~\ref{line:v-dc-start}--\ref{line:v-dc-end} of Algorithm~\ref{alg:vanilla-project}: instead of iterating over all tuples in $\mathcal{R}$, we iterate only over suspect tuples, potentially reducing the number of constraints.

\begin{example}
\label{ex:suspect}
Consider projecting $p_1$ from Table~\ref{tab:comparison} with $\mathcal{I} = \{\text{Type}, \text{Subloc}\}$. For DC1, which compares properties of the same type, suspects must match $p_1.\text{Type} = \text{Condo}$. From Figure~\ref{fig:dataset_dcs}, only $\{d_1, d_2, d_3\}$ are condos, so we can avoid generating a constraint corresponding to $d_4$.
\end{example}

In a sense, the two optimizations we have seen thus far are mutually exclusive: the first builds many constraints but does it at pre-processing time, namely before getting $t$ as input; the second reduces the number of constraints by focusing on instantiations relevant to $t$, yet this can of course be done only when we have $t$ in hand. We can nevertheless combine these two ideas as follows: suspect-set filtering depends only on immutable attribute values, and we thus cache the optimizer's constraints based on these values, and thereby only compute new instantiations when these values change. In particular, this allows reuse within a tuple's perturb-and-project loop (Algorithm~\ref{alg:perturb-project}) since immutable attributes never change.

\paragraph*{Diversity Constraints} A third optimization, which is orthogonal to the first two, relates to the avoidance of generating similar tuples in the perturb-and-project process. To this end, given the set of previously found projections $\mathcal{P} = \{p_1, \ldots, p_{k-1}\}$ and a threshold $\gamma$, we add to the optimizer before Line~\ref{line:v-obj}:
$$\forall p \in \mathcal{P}: |\{A \not\in \mathcal{I} : |t'.A - p.A| > MAD(A,\mathcal{R})\}| \geq \gamma$$
This requires each new projection to differ from every previous one by more than one MAD unit in at least $\gamma$ mutable attributes.
\section{Experiments}
\label{sec:experiments}

We present an experimental evaluation of our solutions. Our main findings, across multiple datasets, are as follows:
\begin{itemize}[leftmargin=*]
    \item {\bf Realism}: Our solutions achieve realism (zero constraint violations), whereas {\em 55.9--100\% of the CFs produced by \DiCE{}  violate at least one constraint}, and some violate many constraints (Table~\ref{tab:neural_constraints}).
    \item {\bf Distance and Diversity}: {\em we are able to achieve comparable distance and diversity of CFs to that achieved by \DiCE{}}, while also achieving realism (under 10\% proximity and ${\sim}3\%$ diversity difference on most datasets, see Figure \ref{fig:neural_quality} and Table \ref{tab:neural_diversity}). Our diversity optimization significantly improves diversity (Table \ref{tab:neural_diversity}).
    \item {\bf Execution Time}: {\em Our optimizations for projection are crucial and outperform the vanilla use of solvers by up to 63$\times$} (Figure~\ref{fig:cf_runtime_and_preproc}). Compared to \DiCE{}, our methods incur
    1.1--10.5$\times$ overhead, as a cost of realism.
\end{itemize}

Additional experiments (Figure~\ref{fig:projection}) show that our solver-based projection fares well compared to optimal results (when available) and alternatives such as closest-instance retrieval. The appendix also evaluates alternative distance functions (Appendix~\ref{app:distance}).

\subsection{Experimental Setup}

We evaluate on four datasets: Adult-Income~\cite{kohavi1996adult}, a demographic dataset for income prediction; NY-Housing~\cite{ny_housing}, containing New York real estate listings for price classification; Tax~\cite{xiao2022fast}, a synthetic dataset common in DC research; and Census-Income~\cite{census_income}, an extended demographic dataset. Characteristics are summarized in Appendix~\ref{app:datasets}. For each dataset, we designate a set of immutable attributes, as we require specifying which attributes cannot be modified; choices per dataset are detailed in Appendix~\ref{app:datasets}. Binary DCs are mined offline using FastADC~\cite{xiao2022fast}; for Adult and NY, we additionally craft unary constraints based on domain knowledge (see Appendix~\ref{app:constraints}). We generate $k=5$ CFs per tuple using weights $w_1 = \frac{2}{3}$ and $w_2 = \frac{1}{3}$ (Equation~\ref{eq:score}), following \DiCE~\cite{mothilal2020explaining}. We use Z3~\cite{de2008z3} as the SMT solver. Results are averaged over 10 tuples per dataset across 10 runs. Beyond \DiCE~\cite{mothilal2020explaining}, we compare against \Vanilla{} (unoptimized SMT projection), \Exhaustive{} (brute-force optimal), and \BestInDataset{} (closest dataset tuple respecting immutable attributes, covering retrieval methods such as NICE~\cite{brughmans2024nice} and FACE~\cite{poyiadzi2020face}), as well as HoloClean~\cite{rekatsinas2017holoclean} (Appendix~\ref{app:holoclean}) and the plausibility-aware methods of CARLA~\cite{pawelczyk2021carla} (Table~\ref{tab:carla}). We omit causal-recourse methods, which require unavailable structural causal models.
Implementation and hardware details appear in Appendix~\ref{app:implementation}.

\subsection{Main Results}

\paragraph{Realism} 
Table~\ref{tab:neural_constraints} quantifies constraint violations. We report the average number of violated DCs per CF (Avg.\ Viol.), unary constraint violations (Unary Viol.), the number of database tuples conflicting with each CF via binary DCs (Tuple Conf.), an inconsistency measure adapted from~\cite{livshits2021properties} and the percentage of CFs violating at least one constraint (Unreal.\ \%). {\em Our solutions yield zero violations, while \DiCE{} produces significant violations}: 100\% of CFs outputted by \DiCE{} are unrealistic for Adult and Tax, 97.1\% for NY, and 55.9\% for Census. These CFs in fact conflict with many tuples: up to 2,300 tuple conflicts per CF. \paragraph{Plausibility-Aware Baselines} We ask whether methods that already model the data distribution produce DC-consistent CFs. We evaluate CARLA's~\cite{pawelczyk2021carla} plausibility-aware methods on Adult: CEM \cite{dhurandhar2018cem}, CCHVAE~\cite{pawelczyk2020learning}, CRUDS~\cite{downs2020cruds}, and FeatureTweak~\cite{tolomei2017interpretable} (GrowingSpheres and CLUE produced no CF). Among domain-valid CFs (with correct one-hot categories and typing), every method, including the VAE-based CCHVAE and CRUDS, violates at least one DC in $100\%$ of cases, with $3.5$ to $4.1$ violations on average (Table~\ref{tab:carla}).

\begin{table}[!htbp]
\centering
\caption{CARLA plausibility-aware methods on Adult. Even generative methods that model the data distribution violate DCs.}
\label{tab:carla}
\small
\setlength{\tabcolsep}{4pt}
\begin{tabular}{|l|l|c|c|c|}
\hline
\textbf{Method} & \textbf{Type} & \textbf{CFs Found} & \textbf{Mean Viol.} & \textbf{Viol. \%} \\
\hline
CEM~\cite{dhurandhar2018cem}                 & gradient & 9/10  & 4.11 & 100 \\
CCHVAE~\cite{pawelczyk2020learning}            & VAE      & 10/10 & 3.90 & 100 \\
CRUDS~\cite{downs2020cruds}                    & CSVAE    & 8/10  & 3.50 & 100 \\
FeatureTweak~\cite{tolomei2017interpretable}  & tree     & 3/10  & 3.67 & 100 \\
\Ours{}                                      & DCs      & 10/10 & 0.00 & 0 \\
\hline
\end{tabular}
\end{table}


\paragraph{Proximity} 
Figure~\ref{fig:neural_quality} compares the proximity achieved by \DiCE{} and by our solutions. Our method, \Ours{}, combines perturb-and-project (Algorithm~\ref{alg:perturb-project}) with the projection optimizations of Section~\ref{sec:project}, including diversity constraints. We also evaluate \Ours{}-Div, a variant without diversity constraints optimization, discussed further in the ablation study. We report two metrics: the combined distance function $dist_{agg}$ (Equation~\ref{eq:dist_agg}), which is the objective both \DiCE{} and our method directly optimize, and L0 distance, counting modified attributes (categorical and numerical). Our methods achieve comparable proximity despite enforcing realism: $dist_{agg}$ increases under 10\% on Adult and Tax, up to ${\sim}31\%$ on Census (9\% without diversity constraints), with similar L0 trends. On NY, we achieve 36--45\% lower $dist_{agg}$ as \DiCE{} overshoots into invalid regions. These gaps to \DiCE{} are statistically significant (paired Wilcoxon, $p<0.001$) but small, reflecting a consistent, modest cost of realism.

\begin{table}[!htbp]
\centering
\caption{Constraint violations for \DiCE. Our solutions yield zero violations.}
\label{tab:neural_constraints}
\small
\setlength{\tabcolsep}{4pt}
\begin{tabular}{|l|r|r|r|r|}
\hline
\textbf{Dataset} & \textbf{Avg. Viol.} & \textbf{Unary Viol.} & \textbf{Tuple Conf.} & \textbf{Unreal. \%} \\
\hline
Adult & 4.16±0.38 & 0.82±0.39 & 778.58±247.03 & 100.0 \\
\hline
NY & 0.71±0.49 & 0.33±0.45 & 22.11±27.04 & 97.1 \\
\hline
Tax & 4.16±1.09 & 0.00±0.00 & 2,326.85±442.89 & 100.0 \\
\hline
Census & 22.18±15.14 & 0.00±0.00 & 722.16±1,101.19 & 55.9 \\
\hline
\end{tabular}
\end{table}

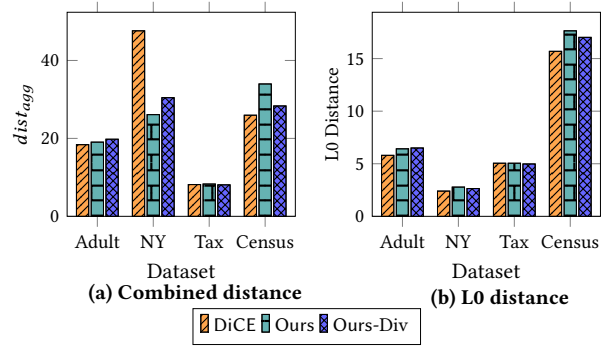
\begin{figure}[!htbp]
\centering
\begin{tikzpicture}[scale = 0.925]
\begin{groupplot}[
    group style={
        group size=2 by 1,
        horizontal sep=1.1cm,
    },
    symbolic x coords={Adult, NY, Tax, Census},
    xtick=data,
    xticklabel style={font=\small},
    yticklabel style={font=\small},
    height=\figheight,
    x=0.8cm,
    ymin=0,
    enlarge x limits=0.18,
]
\nextgroupplot[
    ybar=1pt,
    bar width=5pt,
    ylabel={$dist_{agg}$},
    ylabel style={yshift=-3pt,font=\small},
    xlabel={Dataset},
    legend style={
        at={(1.05,-0.45)},
        anchor=north,
        legend columns=3,
        font=\small,
    },
]
\addplot[fill=orange!70, postaction={pattern=north east lines}] coordinates {(Adult,18.39) (NY,47.61) (Tax,8.15) (Census,25.95)};
\addplot[fill=teal!60, postaction={pattern=bricks}] coordinates {(Adult,19.04) (NY,26.10) (Tax,8.29) (Census,33.97)};
\addplot[fill=blue!60, postaction={pattern=crosshatch}] coordinates {(Adult,19.79) (NY,30.45) (Tax,8.08) (Census,28.33)};
\legend{\DiCE, \Ours, \Ours{}-Div}
\coordinate (bot1) at (rel axis cs:0.5,0);
\nextgroupplot[
    ybar=1pt,
    bar width=5pt,
    ylabel={L0 Distance},
    ylabel style={yshift=-3pt,font=\small},
    xlabel={Dataset},
]
\addplot[fill=orange!70, postaction={pattern=north east lines}] coordinates {(Adult,5.80) (NY,2.40) (Tax,5.05) (Census,15.69)};
\addplot[fill=teal!60, postaction={pattern=bricks}] coordinates {(Adult,6.42) (NY,2.78) (Tax,5.05) (Census,17.64)};
\addplot[fill=blue!60, postaction={pattern=crosshatch}] coordinates {(Adult,6.50) (NY,2.64) (Tax,4.98) (Census,17.02)};
\coordinate (bot2) at (rel axis cs:0.5,0);
\end{groupplot}
\node[below=0.92cm, font=\small] at (bot1) {\textbf{(a) Combined distance}};
\node[below=0.92cm, font=\small] at (bot2) {\textbf{(b) L0 distance}};
\end{tikzpicture}
\caption{CF proximity metrics. Our methods achieve comparable proximity while guaranteeing realism.}
\label{fig:neural_quality}
\end{figure}

\paragraph{Diversity} We measure diversity using the Determinantal Point Process (DPP, Equation~\ref{eq:dpp}). Table~\ref{tab:neural_diversity} shows that \Ours{} achieves DPP scores comparable to \DiCE{}: only 1.7\% lower for Adult, 0.6\% lower for Tax, and identical for Census. For NY, \DiCE{} achieves higher diversity by 7\%, intuitively because the constraints in this dataset limit realistic CF spread. Additional diversity metrics are reported in Appendix~\ref{app:diversity}. Diversity differences are likewise significant ($p<0.001$) yet within a few percent.

\begin{table}[!htbp]
\centering
\caption{CFs DPP diversity scores~\cite{kulesza2012determinantal}.}
\label{tab:neural_diversity}
\small
\setlength{\tabcolsep}{4pt}
\begin{tabular}{|l|c|c|c|c|}
\hline
\textbf{Method} & \textbf{Adult} & \textbf{NY} & \textbf{Tax} & \textbf{Census} \\
\hline
\DiCE & 0.976±0.006 & 0.979±0.010 & 0.886±0.015 & 0.989±0.002 \\
\hline
\Ours & 0.959±0.010 & 0.913±0.049 & 0.881±0.014 & 0.989±0.002 \\
\hline
\Ours-Div & 0.945±0.033 & 0.714±0.353 & 0.870±0.027 & 0.901±0.285 \\
\hline
\end{tabular}
\end{table}

\paragraph{Execution Time} Figure~\ref{fig:cf_runtime_and_preproc}(a) shows CF generation times. \DiCE{} is fastest but produces unrealistic CFs. \PandP{}(\Vanilla{}) is our perturb-and-project framework using the vanilla projection of Algorithm~\ref{alg:vanilla-project}, rebuilding the solver per projection, incurring prohibitive overhead, e.g., 18,871s on Census. Our optimized variants, \Ours{}(\PreProc{}) and \Ours{}(\Suspect{}), dramatically outperform \PandP{}(\Vanilla{}): up to 63$\times$ faster on Census (301s and 334s vs.\ 18,871s). Compared to \DiCE{}, both variants add moderate overhead as the cost of realism: for \Ours{}(\PreProc{}), 2.9$\times$ on NY, 10.5$\times$ on Adult, and 6.9$\times$ on Census (\Ours{}(\Suspect{}) exhibits similar ratios).

The two optimizations offer complementary tradeoffs. \PreProc{} pre-builds all DC instantiations upfront (Figure~\ref{fig:cf_runtime_and_preproc}(b)), with a one-time cost of 1.8s (NY) to 6,168s (Census), after which projections reuse the cached solver. \Suspect{}, by contrast, builds constraints on-demand by filtering to tuples relevant to the specific immutable attribute values (Section~\ref{sec:project}), incurring no upfront cost. On Census, where the full constraint set contains over 7M instantiations, \Suspect{} reduces this by 96.5\% through filtering, achieving comparable runtime (334s vs.\ 301s) without the heavy preprocessing step. \PreProc{} is preferable when generating CFs for many tuples, amortizing its upfront cost; \Suspect{} is preferable when CFs are needed for few tuples or when preprocessing is impractical. A detailed comparison of bounds-building costs and constraint counts is provided in Appendix~\ref{app:suspect_eval}. We further study how runtime scales with dataset size and number of constraints in Appendix~\ref{app:scalability}.

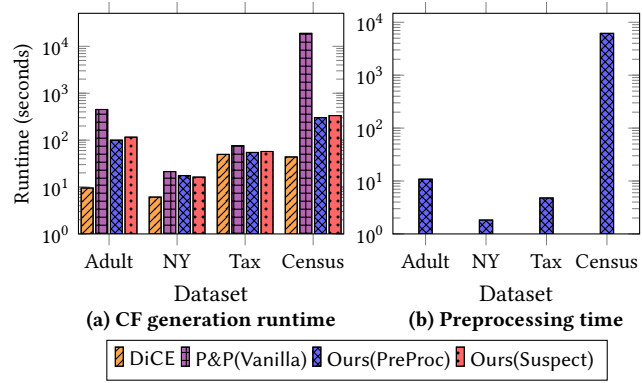
\begin{figure}[!htbp]
\centering
\begin{tikzpicture}
\begin{groupplot}[
    group style={
        group size=2 by 1,
        horizontal sep=0.65cm,  
    },
    symbolic x coords={Adult, NY, Tax, Census},
    xtick=data,
    xticklabel style={font=\small},
    yticklabel style={font=\small},
    height=\figheight,
    ymode=log,
    log origin=infty,
    ymin=1,
    ytick={1,10,100,1000,10000},
    yticklabels={$10^0$,$10^1$,$10^2$,$10^3$,$10^4$},  
]
\nextgroupplot[
    x=0.9cm,  
    ybar=1pt,
    bar width=4.5pt,
    ylabel={Runtime (seconds)},
    ylabel style={yshift=-3pt,font=\small},
    xlabel={Dataset},
    legend style={
        at={(1.05,-0.48)},
        anchor=north,
        legend columns=4,
        font=\small,
    },
    enlarge x limits=0.15,
]
\addplot[fill=orange!70, postaction={pattern=north east lines}] coordinates {(Adult,9.53) (NY,6.10) (Tax,49.61) (Census,43.77)};
\addplot[fill=violet!60, postaction={pattern=grid}] coordinates {(Adult,449.71) (NY,21.33) (Tax,76.06) (Census,18871.76)};
\addplot[fill=blue!60, postaction={pattern=crosshatch}] coordinates {(Adult,99.85) (NY,17.59) (Tax,54.59) (Census,301.27)};
\addplot[fill=red!60, postaction={pattern=dots}] coordinates {(Adult,115.56) (NY,16.27) (Tax,57.21) (Census,333.61)};
\legend{\DiCE{}, \PandP{}(\Vanilla{}), \Ours(\PreProc{}), \Ours(\Suspect{})}
\coordinate (bot1) at (rel axis cs:0.5,0);
\nextgroupplot[
    x=0.8cm,  
    ybar=1pt,
    bar width=5pt,
    xlabel={Dataset},
    enlarge x limits=0.18,
]
\addplot[fill=blue!60, postaction={pattern=crosshatch}] coordinates {(Adult,10.83) (NY,1.83) (Tax,4.78) (Census,6168.02)};
\coordinate (bot2) at (rel axis cs:0.5,0);
\end{groupplot}
\node[below=0.9cm, font=\small] at (bot1) {\textbf{(a) CF generation runtime}};
\node[below=0.9cm, font=\small] at (bot2) {\textbf{(b) Preprocessing time}};
\end{tikzpicture}
\caption{(a) \Ours{} guarantees realistic CFs with up to 63$\times$ speedup over \Vanilla. (b) \PreProc{} preprocessing runtime cost.}
\label{fig:cf_runtime_and_preproc}
\end{figure}

\vspace{-1mm}
\subsection{Ablation}

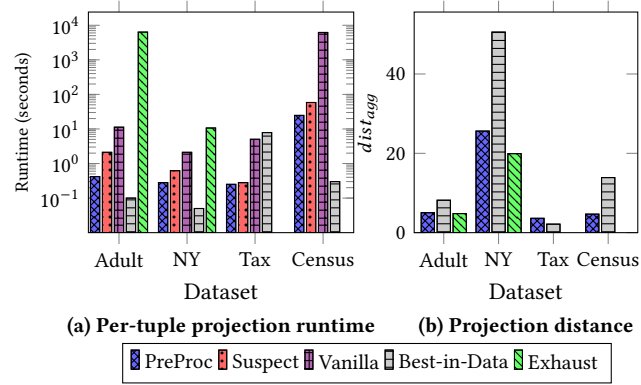
\begin{figure}[!htbp]
\centering
\begin{tikzpicture}
\begin{groupplot}[
    group style={
        group size=2 by 1,
        horizontal sep=0.8cm,
    },
    symbolic x coords={Adult, NY, Tax, Census},
    xtick=data,
    xticklabel style={font=\small},
    yticklabel style={font=\small},
    height=\figheight,
]
\nextgroupplot[
    x=0.9cm,  
    ybar=1pt,
    bar width=3.5pt,
    ymin=0.01,
    ymode=log,
    log origin=infty,
    ytick={0.1,1,10,100,1000,10000},
    yticklabels={$10^{-1}$,$10^0$,$10^1$,$10^2$,$10^3$,$10^4$},
    ylabel={Runtime (seconds)},
    ylabel style={yshift=-3pt,font=\footnotesize},
    xlabel={Dataset},
    legend style={
        at={(1.05,-0.50)},
        anchor=north,
        legend columns=5,
        font=\small,
    },
    enlarge x limits=0.15,
]
\addplot[fill=blue!60, postaction={pattern=crosshatch}] coordinates {(Adult,0.42) (NY,0.28) (Tax,0.25) (Census,24.68)};
\addplot[fill=red!60, postaction={pattern=dots}] coordinates {(Adult,2.11) (NY,0.62) (Tax,0.28) (Census,58.33)};
\addplot[fill=violet!60, postaction={pattern=grid}] coordinates {(Adult,11.25) (NY,2.11) (Tax,5.03) (Census,6192.70)};
\addplot[fill=gray!40, postaction={pattern=horizontal lines}] coordinates {(Adult,0.10) (NY,0.05) (Tax,7.86) (Census,0.30)};
\addplot[fill=green!60, postaction={pattern=north west lines}] coordinates {(Adult,6344) (NY,10.66)};
\legend{PreProc, Suspect, Vanilla, Best-in-Data, Exhaust}
\coordinate (bot1) at (rel axis cs:0.5,0);
\nextgroupplot[
    x=0.725cm,
    ybar=1pt,
    bar width=5pt,
    ymin=0,
    ylabel={$dist_{agg}$},
    ylabel style={yshift=-3pt,font=\footnotesize},
    xlabel={Dataset},
    enlarge x limits=0.18,
]
\addplot[fill=blue!60, postaction={pattern=crosshatch}] coordinates {(Adult,5.03) (NY,25.61) (Tax,3.62) (Census,4.67)};
\addplot[fill=gray!40, postaction={pattern=horizontal lines}] coordinates {(Adult,8.18) (NY,50.56) (Tax,2.15) (Census,13.88)};
\addplot[fill=green!60, postaction={pattern=north west lines}] coordinates {(Adult,4.81) (NY,19.92)};
\coordinate (bot2) at (rel axis cs:0.5,0);
\end{groupplot}
\node[below=1.0cm, font=\small] at (bot1) {\textbf{(a) Per-tuple projection runtime}};
\node[below=1.0cm, font=\small] at (bot2) {\textbf{(b) Projection distance}};
\end{tikzpicture}
\caption{(a) Per-tuple runtime: \PreProc{} is fastest after preprocessing; \Exhaustive{} feasible only for Adult and NY. (b) Projection distance: \PreProc{} achieves near-optimal quality.}
\label{fig:projection}
\vspace{-3.5mm}
\end{figure}

\paragraph{Diversity Optimization} 
Table~\ref{tab:neural_diversity} demonstrates the effect of our diversity constraints, with $\gamma = 2$. \Ours{} consistently improves diversity over \Ours{}-Div, most notably on NY, where the DPP score increases from 0.714 to 0.913. Without diversity constraints, the tight feasible region in NY leads to clustered projections. With diversity constraints, \Ours{} achieves DPP scores comparable to \DiCE{}: within 1.7\% on Adult, 0.6\% on Tax, and near identical on Census. In terms of proximity, \Ours{} remains comparable to both \Ours{}-Div and \DiCE{} (Figure~\ref{fig:neural_quality}), with less than 4\% difference. The diversity optimization incurs negligible runtime overhead (omitted from Figure~\ref{fig:projection}).
\paragraph{Alternative Projection Methods} Figure~\ref{fig:projection}(a) evaluates projection strategies in isolation. Using a solver without our optimizations (\Vanilla{}) is prohibitively slow: 8--27$\times$ slower than \PreProc{} on Adult, NY, and Tax, and 250$\times$ slower on Census. We also evaluate non-solver alternatives: \Exhaustive{} (enumerating the discretized valid space for optimal projections) and \BestInDataset{} (returning the closest dataset tuple matching immutable attributes). \Exhaustive{} is impractical, exceeding 100 minutes per projection on Adult and 12-hour limits on Tax and Census. \BestInDataset{} has 10--20\% failure rates when immutable attributes are present, and among successful projections, 63--197\% higher distance on Adult, NY, and Census (Figure~\ref{fig:projection}(b)). On Tax, where no attributes are immutable, \BestInDataset{} achieves lower distance but is significantly slower. We further evaluated HoloClean~\cite{rekatsinas2017holoclean}, a probabilistic data cleaning framework; it leaves 89\% (Adult) and 100\% (NY) of violations unresolved (Appendix~\ref{app:holoclean}), confirming general-purpose cleaning is unsuitable for targeted projection.

\section{Related Work}
\label{sec:related}

CF generation has been extensively studied in Explainable AI~\cite{wachter2017counterfactual,jiang2024robust,verma2024counterfactual,guidotti2024counterfactual}. A related but distinct concept, \emph{algorithmic recourse}~\cite{ustun2019actionable,karimi2021algorithmic}, specifies \emph{how} to reach a different prediction through actionable interventions, whereas CF explanations identify \emph{what} features would need to differ. Our work focuses on CF explanations, finding realistic alternative instances, rather than prescribing action sequences.

Most lines of research~\cite{bewley2024counterfactual,dandl2020multi} in this context, such as DiCE which we have discussed above, do not address the realism of resulting CFs. Our experiments with DiCE demonstrate that realism is unlikely to be achieved spuriously if one does not optimize for it. Other approaches do aim at realism, there is no single notion of realism, and they capture different, partially overlapping facets of it. Some approaches~\cite{poyiadzi2020face,brughmans2024nice,jiang2025robustx} limit attention to instances in a given dataset, e.g., the training set. As we have demonstrated (see comparisons with best-in-dataset in Section~\ref{sec:experiments}), this may hamper quality as the dataset may simply not include sufficiently many, close, or diverse CF candidates. There is also a privacy concern: explaining a prediction made for Bob through the classification of Alice may infringe Alice's right to privacy. A further family of methods captures realism statistically, as closeness to a learned data distribution, e.g.\ via variational autoencoders or prototypes, as in the methods of the CARLA library~\cite{pawelczyk2021carla,dhurandhar2018cem,pawelczyk2020learning,downs2020cruds,tolomei2017interpretable}. To generalize beyond the existing dataset, causal approaches~\cite{karimi2021algorithmic,mahajan2019preserving,karimi2020algorithmic} assume either full structural causal models (SCMs), or a simple causal graph. Both are rarely available in practice and require significant domain knowledge. Crucially, none of these approaches \emph{guarantees} a valid CF: each optimizes closeness to the data, yet a statistically plausible instance may still violate the integrity rules the data obeys, as we confirm experimentally. By contrast, Denial Constraints that we use here (1) may be automatically mined (as we did in our experiments using FastADC~\cite{xiao2022fast}) and (2) were shown in~\cite{ben2024cafa}, in the different context of adversarial training, to capture well the space of valid instances. Indeed adversarial examples somewhat resemble CFs, yet they seek a single example, with no repeated constraint solving (thus no optimization) and no notion of diversity. The only other works, to our knowledge, that incorporate database-style constraints for CF explanations are GeCo~\cite{schleich2021geco} and CeC~\cite{deutch2019constraints}. Both CeC and GeCo support unary but not binary DCs, which significantly hampers expressiveness, and again do not address diversity.

Beyond CFs, constraints~\cite{abiteboul1995foundations} are widely used to capture data integrity, and various formalisms have been proposed to capture constraints. These include functional dependencies~\cite{codd1970relational}, conditional functional dependencies~\cite{bohannon2006conditional,fan2008conditional}, and others. DCs generalize these and other formalisms~\cite{chu2013discovering,bleifuss2017hydra}. Among them, DCs are attractive on several counts. As a \emph{formalism}, each DC is an explicit logical rule, so whether a tuple satisfies it can be checked exactly, unlike a learned model whose validity notion is implicit. DCs are also the most \emph{expressive} such formalism, and are accordingly adopted by most major data-quality systems~\cite{rekatsinas2017holoclean,dallachiesa2013nadeef,geerts2013llunatic}. This is what lets us \emph{guarantee} realism: a tuple either satisfies all DCs or it does not, so enforcing them in the projection yields zero violations by construction. The guarantee is also robust to imperfect constraints, as DCs can be mined automatically with perfect recall~\cite{chu2013discovering,bleifuss2017hydra,xiao2022fast}, and an overly specific mined set only narrows the admissible space, never admitting an invalid CF, while a domain expert may add further constraints. Beyond data cleaning, DCs are also used in consistent query answering~\cite{bertossi2019database,arenas1999consistent}, data integration, and data profiling~\cite{abedjan2015profiling}. Constraints are also extensively used in the context of data repair, where the goal is to find (minimal) modifications to the database so that a set of constraints holds~\cite{datarepairsurvey,chu2013holistic,rekatsinas2017holoclean,song2016constraint}. While our suspect-set optimization is inspired by data cleaning literature~\cite{chu2013holistic}, a key difference is that we focus on modifying a single tuple rather than the entire database. Indeed, we have shown that using HoloClean~\cite{rekatsinas2017holoclean} for this purpose fails.
\section{Conclusions}
\label{sec:conc}

We have studied in this paper the problem of finding counterfactuals for predictions made by ML models, constraining that the counterfactuals are {\em realistic} as captured by Denial Constraints with respect to a given database. We optimize for proximity and diversity in this constrained space. We have shown that the state-of-the-art solution often yields non-realistic counterfactuals, whereas we achieve 100\% realism and comparable quality with respect to proximity and diversity. We show that a vanilla constraint solver-based solution is not scalable, yet we are able to achieve scalability through dedicated optimizations. For future work, we will explore additional optimizations as well as implementations in additional application domains such as that of adversarial robustness. 

\begin{acks}
This work of  Daniel Deutch and Avia Asael was partially funded by the European Research Council (ERC) under the European Union's Horizon 2020 research and innovation programme (Grant agreement No. 804302), by the Israeli Science Foundation (Grant 1476/24), by Len Blavatnik and the Blavatnik Family foundation, and by the Deutsch foundation. The work of Amir Gilad was funded by the Israel Science Foundation (ISF) under grant 1702/24, the Scharf-Ullman Endowment, and the Alon Scholarship.
\end{acks}
\clearpage
\bibliographystyle{ACM-Reference-Format}
\balance
\bibliography{sample}

\clearpage
\appendix
\section{Proofs}
\label{app:proofs}

\begin{proof}[Proof of Proposition~\ref{prop:np-hard}]
The proof is via reduction from 3-SAT. Given a boolean CNF expression $\phi$, we design a relation $\mathcal{R}$ whose attributes correspond to the variables of $\phi$. For every clause $C_i = l_1 \vee l_2 \vee l_3$ in $\phi$, we design the DC:
$\forall t.\ \neg (t.a_1 = b_1 \wedge t.a_2 = b_2 \wedge t.a_3 = b_3)$
where $a_i$ is the variable of literal $l_i$, and $b_i = 1$ if $a_i$ appears negatively in $l_i$, otherwise $b_i = 0$. 

We have that $\phi$ is satisfiable if and only if there exists a tuple satisfying all DCs. For one direction, consider a satisfying assignment to $\phi$, and design a tuple with $1$ in attributes corresponding to variables set to $true$, and $0$ elsewhere. Since the assignment satisfies every clause, the tuple satisfies every DC. Conversely, a tuple satisfying all DCs translates into a truth assignment satisfying all clauses.
\end{proof}

\section{Additional Experimental Details}
\label{app:datasets}

\begin{table}[h]
  \centering
  \caption{Dataset characteristics: number of tuples, attributes, and DCs.}
  \label{tab:datasets_desc}
  \small
  \begin{tabular}{l r @{\hspace{10pt}} r @{\hspace{10pt}} r r @{\hspace{10pt}} r}
    \toprule
    \multirow{2}{*}{Dataset} & \multirow{2}{*}{Tuples} & \multirow{2}{*}{Attributes} & \multicolumn{2}{c}{Attr. Types} & \multirow{2}{*}{Constraints} \\
    \cmidrule(lr){4-5}
    & & & Cat & Num & \\
    \midrule
    Adult & 30,014 & 11 & 8 & 3 & 6 \\
    NY-Housing & 2,996 & 6 & 3 & 3 & 9 \\
    Tax & 99,904 & 9 & 5 & 4 & 8 \\
    Census & 94,786 & 39 & 28 & 11 & 200 \\
    \bottomrule
  \end{tabular}
\end{table}

The Adult-Income dataset~\cite{kohavi1996adult} contains demographic and educational attributes; the task is to predict whether income exceeds \$50,000. The NY-Housing dataset~\cite{ny_housing} contains New York real estate properties; the task is to classify whether price exceeds \$1,000,000. The Tax dataset~\cite{xiao2022fast} is synthetic and commonly used in DC discovery research; the task is to predict whether tax rate exceeds 5\%. The Census-Income dataset~\cite{census_income} is similar to Adult with more recent data and richer attributes.

\subsection{Immutable Attributes}
\label{app:immutables}
Immutable attributes determine which attributes CFs are allowed to modify, focusing the search on relevant changes. For each dataset, we designated a representative set: demographic attributes such as age, race, and sex for Adult and Census, and structural attributes such as property type and sublocality for NY-Housing. These choices reflect attributes that are naturally fixed in each domain; other reasonable choices are possible and our framework supports any such designation. For Tax, no attributes were designated as immutable.

\subsection{Constraint Selection}
\label{app:constraints}

For Adult and NY, we selected meaningful binary constraints from those mined by FastADC and manually crafted unary constraints based on domain knowledge reflected in the data. For Tax, we used all mined constraints. For Census, which produced 2,566 constraints, we randomly sampled 200 to maintain tractability while preserving constraint diversity. Example DCs appear in Table~\ref{tab:example_constraints}.

\subsection{Implementation Details}
\label{app:implementation}

For each dataset, we train a neural network classifier with a single hidden layer (100 units, ReLU activation, softmax output) implemented in PyTorch. All experiments ran on a server with Intel Xeon Gold 6252 CPU (2.10 GHz, 96 cores) and 1TB RAM, running Ubuntu 18.04 LTS.

\section{Additional Experiments}
\label{app:experiments}

\subsection{Projection Scalability}
\label{app:scalability}
\begin{figure}[!htbp]
    \centering
    \begin{subfigure}{0.495\columnwidth}
        \centering
        \includegraphics[width=\textwidth]{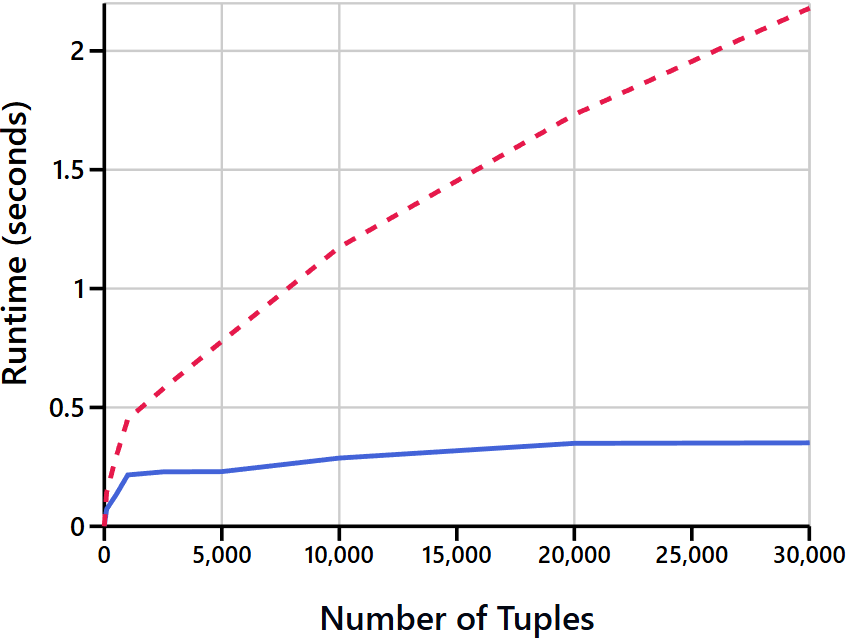}
        \caption{Runtime comparison}
        \label{fig:runtime_dataset_size}
    \end{subfigure}
    \hfill
    \begin{subfigure}{0.495\columnwidth}
        \centering
        \includegraphics[width=\textwidth]{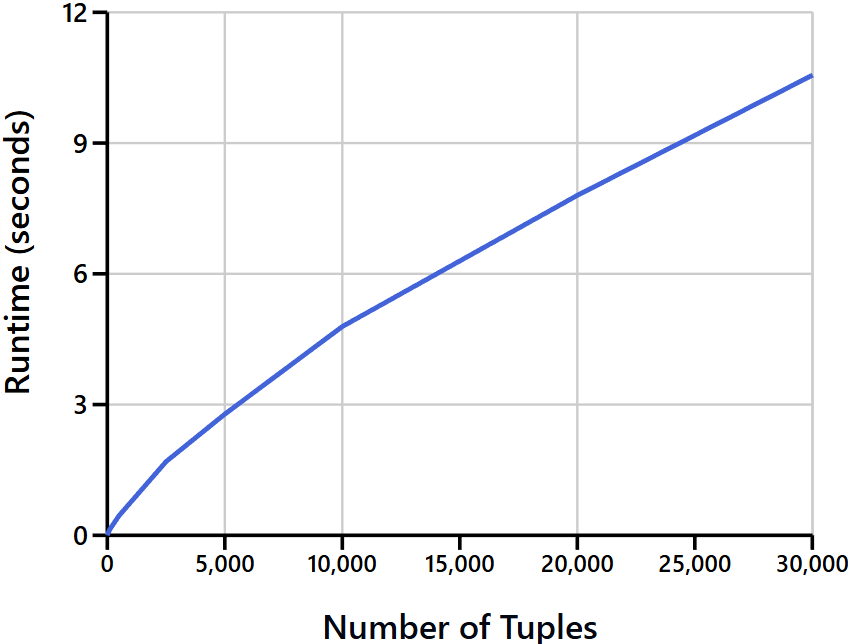}
        \caption{Pre-processing (\textsf{PreProc})}
        \label{fig:preprocessing_dataset_size}
    \end{subfigure}
    
    \vspace{-1mm}
    \centering
    \footnotesize
    \begin{tikzpicture}[baseline=-0.5ex]
        \draw[blue, line width=1.5pt] (0,0) -- (0.4,0);
    \end{tikzpicture} PreProc \quad
    \begin{tikzpicture}[baseline=-0.5ex]
        \draw[red, line width=1.5pt, dashed] (0,0) -- (0.4,0);
    \end{tikzpicture} Suspect
    \vspace{-2mm}
    \caption{Scalability w.r.t.\ dataset size (Adult dataset).}
    \label{fig:dataset_size_combined}
\end{figure}
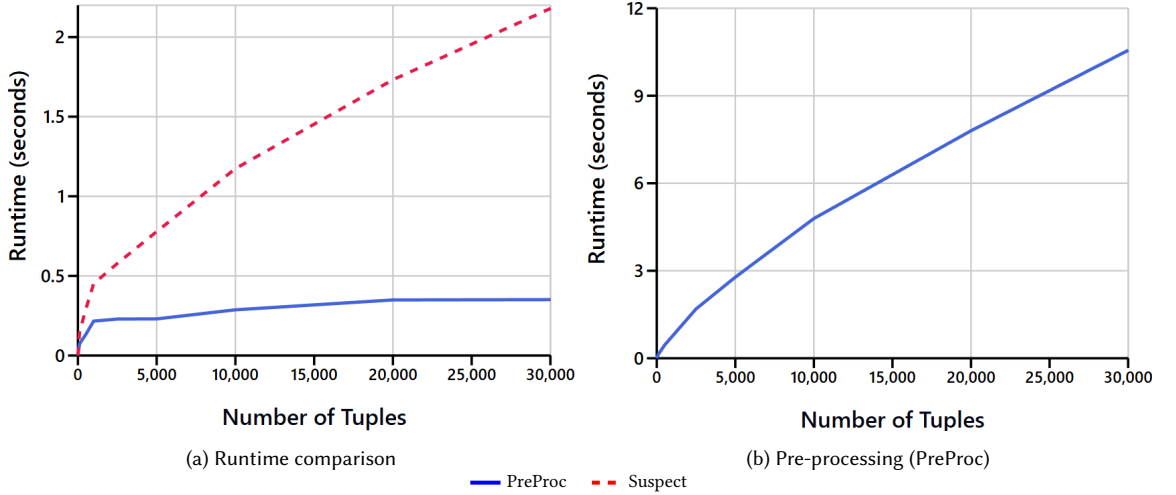

\paragraph{Runtime vs Dataset Size}
Runtime scales linearly with dataset size for both variants. Figure~\ref{fig:dataset_size_combined} shows Adult results: \textsf{PreProc} completes projections in 0.35s at 30K tuples while \textsf{Suspect} requires 2.18s. The linear scaling stems from constraint instantiation: each database tuple potentially participates in DC violations, requiring assertions proportional to relation size. Pre-processing time and memory follow the same linear pattern. Notably, \textsf{Suspect} consumes significantly less memory than \textsf{PreProc} despite similar runtime costs, as it filters constraints at projection time.

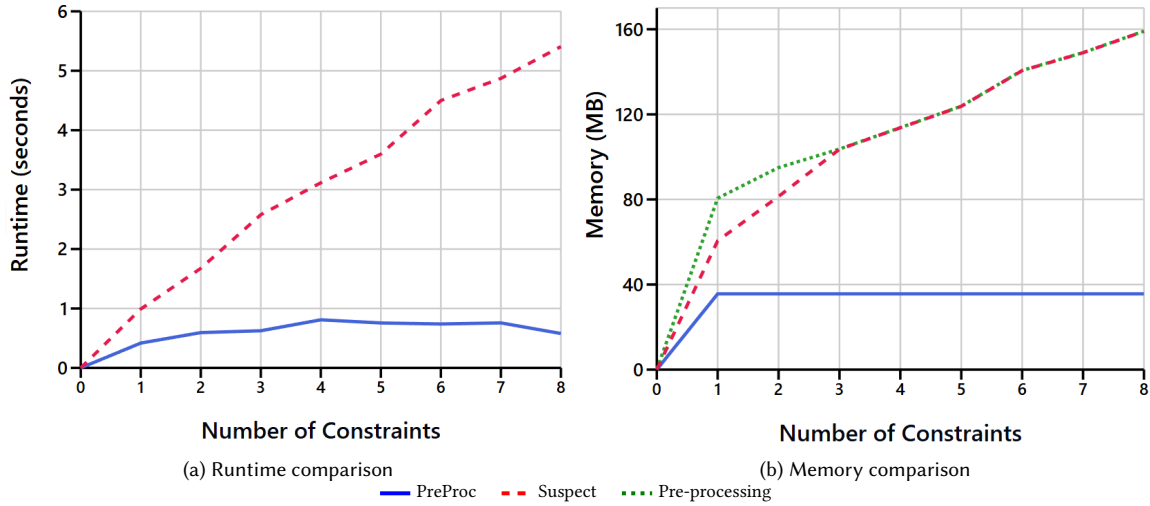
\begin{figure}[!htbp]
    \centering
    \begin{subfigure}{0.495\columnwidth}
        \centering
        \includegraphics[width=\textwidth]{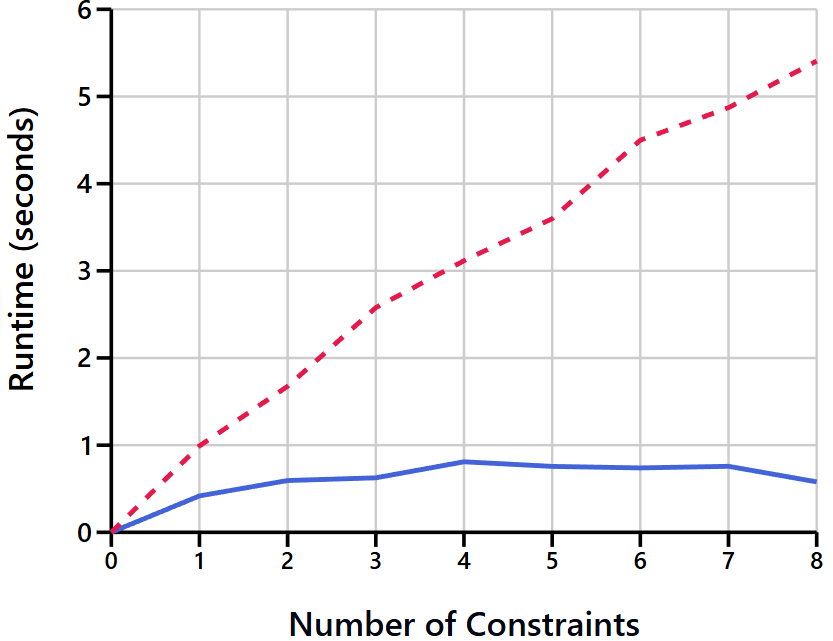}
        \caption{Runtime comparison}
        \label{fig:runtime_constraint_size}
    \end{subfigure}
    \hfill
    \begin{subfigure}{0.495\columnwidth}
        \centering
        \includegraphics[width=\textwidth]{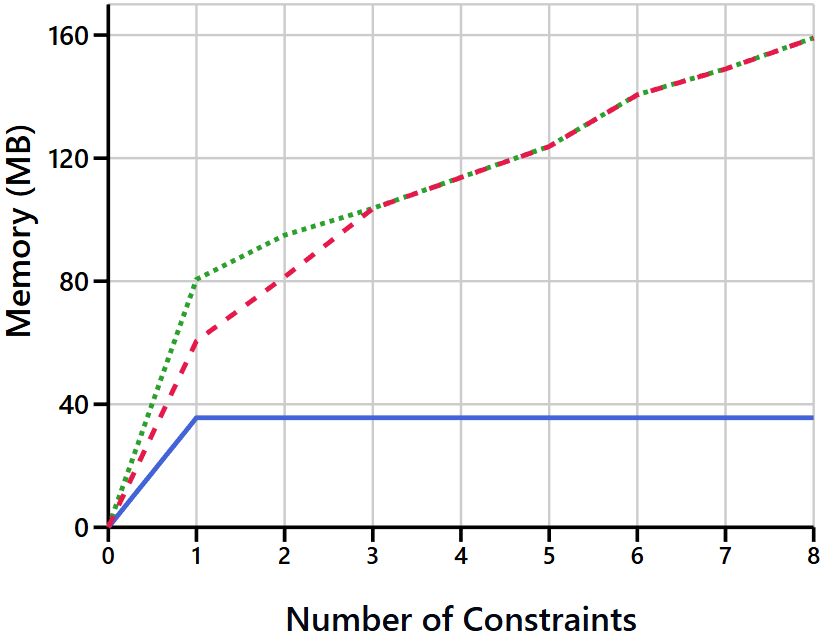}
        \caption{Memory comparison}
        \label{fig:memory_constraint_size}
    \end{subfigure}
    
    \vspace{-1mm}
    \centering
    \footnotesize
    \begin{tikzpicture}[baseline=-0.5ex]
    \draw[blue, line width=1.5pt] (0,0) -- (0.4,0);
    \end{tikzpicture} PreProc \quad
    \begin{tikzpicture}[baseline=-0.5ex]
        \draw[red, line width=1.5pt, dashed] (0,0) -- (0.4,0);
    \end{tikzpicture} Suspect \quad
    \begin{tikzpicture}[baseline=-0.5ex]
        \draw[green!50!black, line width=1.5pt, dotted] (0,0) -- (0.4,0);
    \end{tikzpicture} Pre-processing
    \vspace{-2mm}
    \caption{Scalability w.r.t.\ number of constraints (Tax dataset).}
    \label{fig:constraints_size_combined}
\end{figure}

\paragraph{Runtime vs Number of Constraints}
Runtime increases linearly with the number of denial constraints, as demonstrated on Tax (Figure~\ref{fig:constraints_size_combined}). \textsf{PreProc} projection time grows from 0.42s with 1 constraint to 0.58s with 8 constraints, while \textsf{Suspect} increases from 0.99s to 5.41s. Memory scales similarly, from 80MB to 159MB for pre-processing. \textsf{Suspect} shows steeper runtime growth because it cannot amortize constraint compilation, rebuilding the optimizer for each projection.

\subsection{Additional Diversity Metrics}
\label{app:diversity}

Table~\ref{tab:diversity_full} reports additional diversity metrics beyond the DPP scores presented in the main paper. Average pairwise distance (Pair.\ Div.) measures the mean $dist_{agg}$ between all CF pairs, capturing overall spread. Minimum pairwise distance (Min.\ Dist.\ Div.) reports the smallest distance between any two CFs, indicating whether any pair is overly similar~\cite{mothilal2020explaining,dandl2020multi}. Both are computed using $dist_{agg}$ (Equation~\ref{eq:dist_agg}).

The trends are consistent with the DPP results: \Ours{} achieves pairwise diversity comparable to \DiCE{} on most datasets, with the largest gap on NY due to the tight feasible region imposed by constraints. Our framework currently optimizes for DPP diversity; adapting the diversity constraints in Section~\ref{sec:project} to directly optimize alternative diversity measures such as pairwise or minimum distance is a natural direction for future work.

\begin{table}[h]
\centering
\caption{Full diversity metrics. DPP: Determinantal Point Process diversity~\cite{kulesza2012determinantal}; Pair.\ Div.: average pairwise distance; Min.\ Dist.\ Div.: minimum distance between any CF pair.}
\label{tab:diversity_full}
\small
\setlength{\tabcolsep}{4pt}
\begin{tabular}{|l|l|c|c|c|}
\hline
\textbf{Dataset} & \textbf{Method} & \textbf{DPP} & \textbf{Pair.\ Div.} & \textbf{Min.\ Dist.\ Div.} \\
\hline
\multirow{3}{*}{Adult} 
 & \DiCE & 0.976±0.006 & 21.71±3.49 & 12.33±2.37 \\
 & \Ours & 0.959±0.010 & 16.13±2.27 & 9.00±2.62 \\
 & \Ours-Div & 0.945±0.033 & 15.95±2.67 & 8.03±3.34 \\
\hline
\multirow{3}{*}{NY} 
 & \DiCE & 0.979±0.010 & 38.64±10.53 & 10.76±5.14 \\
 & \Ours & 0.913±0.049 & 21.40±8.22 & 5.05±2.68 \\
 & \Ours-Div & 0.714±0.353 & 21.34±11.95 & 3.23±2.35 \\
\hline
\multirow{3}{*}{Tax} 
 & \DiCE & 0.886±0.015 & 8.23±0.69 & 5.00±0.95  \\
 & \Ours & 0.881±0.014 & 7.99±0.56 & 4.91±1.00 \\
 & \Ours-Div & 0.870±0.027 & 7.99±0.54 & 4.18±1.27 \\
\hline
\multirow{3}{*}{Census} 
 & \DiCE & 0.989±0.002 & 28.76±2.23 & 20.97±2.87 \\
 & \Ours & 0.989±0.002 & 29.45±3.09 & 20.60±2.42 \\
 & \Ours-Div & 0.901±0.285 & 31.88±3.25 & 20.33±6.85 \\
\hline
\end{tabular}
\end{table}

\subsection{Suspect-Set Optimization Evaluation}
\label{app:suspect_eval}

Figure~\ref{fig:suspect_eval} evaluates the Suspect-set filtering optimization described in Section~\ref{sec:project}. Figure~\ref{fig:suspect_eval}(a) compares bounds-building time: \Suspect{} reduces this step by 6--35$\times$ compared to \PreProc{}, with the largest gain on Census (176.8s vs.\ 6,168s). For Tax, both methods produce identical bounds since all tuples are suspects. Figure~\ref{fig:suspect_eval}(b) shows the number of optimizer bounds (DC instantiations). \Suspect{} reduces instantiations by up to 96.5\% on Census (246K vs.\ 7.19M), explaining the runtime gains. On datasets with fewer tuples or less selective immutable attributes, the reduction is smaller but still substantial (e.g., 85\% on Adult and NY).

\begin{figure}[!htbp]
\centering
\begin{tikzpicture}
\begin{groupplot}[
    group style={
        group size=2 by 1,
        horizontal sep=0.9cm,
    },
    symbolic x coords={Adult, NY, Tax, Census},
    xtick=data,
    xticklabel style={font=\small},
    yticklabel style={font=\small},
    height=\figheight,
    ymode=log,
    log origin=infty,
    enlarge x limits=0.18,
]
\nextgroupplot[
    x=0.8cm,
    ybar=1pt,
    bar width=5pt,
    ymin=0.5,
    ytick={1,10,100,1000,10000},
    yticklabels={$10^0$,$10^1$,$10^2$,$10^3$,$10^4$},
    ylabel={Time (seconds)},
    ylabel style={yshift=-3pt,font=\small},
    xlabel={Dataset},
    legend style={
        at={(1.05,-0.45)},
        anchor=north,
        legend columns=2,
        font=\small,
    },
]
\addplot[fill=blue!60, postaction={pattern=crosshatch}] coordinates {(Adult,10.83) (NY,1.83) (Tax,4.78) (Census,6168.02)};
\addplot[fill=red!60, postaction={pattern=dots}] coordinates {(Adult,1.86) (NY,1.14) (Tax,4.78) (Census,176.78)};
\legend{\PreProc{}, \Suspect{}}
\coordinate (bot1) at (rel axis cs:0.5,0);
\nextgroupplot[
    x=0.8cm,
    ybar=1pt,
    bar width=5pt,
    ymin=100,
    ytick={100,1000,10000,100000,1000000,10000000},
    yticklabels={$10^2$,$10^3$,$10^4$,$10^5$,$10^6$,$10^7$},
    ylabel={\# Optimizer Bounds},
    ylabel style={yshift=-3pt,font=\small},
    xlabel={Dataset},
]
\addplot[fill=blue!60, postaction={pattern=crosshatch}] coordinates {(Adult,25468) (NY,4620) (Tax,1143) (Census,7189517)};
\addplot[fill=red!60, postaction={pattern=dots}] coordinates {(Adult,3597) (NY,622) (Tax,1143) (Census,246563)};
\coordinate (bot2) at (rel axis cs:0.5,0);
\end{groupplot}
\node[below=0.92cm, font=\small] at (bot1) {\textbf{(a) Bounds-building time}};
\node[below=0.92cm, font=\small] at (bot2) {\textbf{(b) Number of bounds}};
\end{tikzpicture}
\caption{Suspect-set filtering evaluation. (a)~\Suspect{} reduces bounds-building time by up to 35$\times$. (b)~\Suspect{} reduces the number of DC instantiations by up to 96.5\%.}
\label{fig:suspect_eval}
\end{figure}
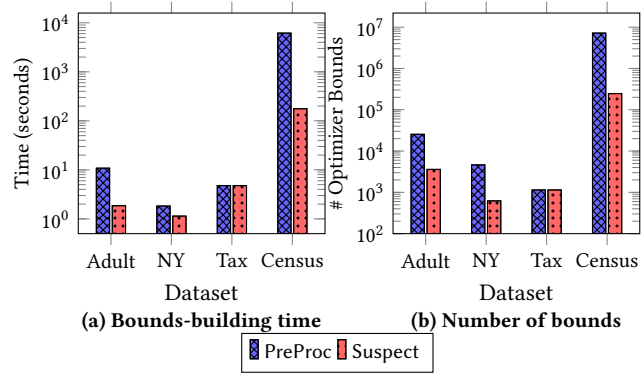

\subsection{Projection via HoloClean}
\label{app:holoclean}

We evaluated HoloClean~\cite{rekatsinas2017holoclean}, a probabilistic data cleaning framework, as an alternative for projection. Table~\ref{tab:holoclean} shows that HoloClean leaves 89\% (Adult) and 100\% (NY) violations unresolved, confirming general-purpose cleaning is unsuitable for targeted tuple projection.

\begin{table}[!htbp]
\centering
\caption{HoloClean projection results. Violations remain largely unresolved.}
\label{tab:holoclean}
\small
\setlength{\tabcolsep}{2.25pt}
\begin{tabular}{|l|l|c|c|c|c|}
\hline
\textbf{Dataset} & \textbf{Method} & \textbf{Avg. Viol.} & \textbf{Unary Viol.} & \textbf{Tuple Conf.} & \textbf{Unreal. \%} \\
\hline
\multirow{2}{*}{Adult} & Input & 4.18±0.71 & 0.63±0.48 & 575.27±538.55 & 100.0 \\
 & HC & 3.72±1.05 & 0.63±0.48 & 366.09±476.22 & 89.0 \\
 \hline
\multirow{2}{*}{NY} & Input & 1.40±1.11 & 1.00±1.00 & 77.00±229.33 & 100.0 \\
 & HC & 1.40±1.11 & 1.00±1.00 & 77.00±229.33 & 100.0 \\
\hline
\end{tabular}
\end{table}

\subsection{Custom Distance Functions}
\label{app:distance}

Our projection algorithm supports arbitrary distance functions encoded in the solver. We evaluate three distance functions on the Adult dataset: L0 (counting changed attributes for both categorical and numerical), $dist_{agg}$ (Equation~\ref{eq:dist_agg}), and a custom function with domain-specific categorical distances. The custom function encodes semantic transition costs; for example, changing occupation from Farming\_fishing to Prof\_specialty has distance 4 (reflecting a significant career change), while transitioning marital status from Widowed to Never\_married is assigned $\infty$ (impossible).

\begin{table}[!htbp]
\centering
\caption{Projection with different distance functions on Adult. Each row optimizes the distance function in the first column; remaining columns show the resulting projection distances under all three metrics.}
\label{tab:distance_functions}
\small
\setlength{\tabcolsep}{4pt}
\begin{tabular}{|l|r|r|r|r|}
\hline
\textbf{Optimized} & \textbf{Runtime (s)} & \textbf{$dist_{agg}$} & \textbf{L0} & \textbf{Custom} \\
\hline
L0 & 0.17 & 8.82 & 2.55 & 191.03 \\
$dist_{agg}$ & 0.35 & 5.21 & 2.55 & 96.54 \\
Custom & 1.13 & 5.21 & 2.55 & 5.68 \\
\hline
\end{tabular}
\end{table}

Table~\ref{tab:distance_functions} shows that each function successfully optimizes its target metric. Runtime scales with distance function complexity: L0 (simple counting) is fastest at 0.17s, $dist_{agg}$ (MAD-normalized) requires 0.35s, and the custom function (pairwise categorical lookups) takes 1.13s. These results demonstrate that our solver-based projection is robust to the choice of distance function, with runtime proportional to encoding complexity.

\end{document}